\documentclass[12pt]{amsart}
\usepackage{epsfig}
\usepackage{amsfonts,dsfont}
\usepackage{amsmath,amssymb}
\usepackage{verbatim}
\newtheorem{thm}{Theorem}[section]

\newtheorem{cor}[thm]{Corollary}
\newtheorem{lem}[thm]{Lemma}
\newtheorem{prop}[thm]{Proposition}
\newtheorem*{prob*}{Problem}
\newtheorem*{thm*}{Theorem}

\theoremstyle{definition}
\newtheorem{defn}[thm]{Definition}

\newtheorem*{defn*}{Definition}
\newtheorem{rem}[thm]{Remark}

\newtheorem*{rem*}{Remark}
\numberwithin{equation}{section}

\newcommand{\Ga}{\Gamma}

\newcommand{\ka}{\kappa}

\newcommand{\im}{\mathop{\mathrm{Im}}}

\newcommand{\ta}{\tilde{a}}
\newcommand{\tb}{\tilde{b}}
\newcommand{\tE}{\tilde{E}}

\renewcommand{\Im}{\operatorname{\rm Im}}       
\renewcommand{\Re}{\operatorname{\rm Re}}       
\AtBeginDocument{\renewcommand{\d}[1]{\mathrm{d}#1}}
\def\le{\left(}
\def\ri{\right)}
\def\no{\nonumber}
\newcommand{\Eint}{\int_{E}}
\def\G{\Gamma}
\newcommand{\be}{\begin{equation}}
\newcommand{\ee}{\end{equation}}

\begin{document}
\title[Gegenbauer and other orthogonal polynomials on an ellipse]
 {\bf{Gegenbauer and other planar orthogonal poly-\\nomials on an ellipse in the complex plane}}

\author[G. Akemann]{Gernot Akemann}
\address{Faculty of Physics, 
Bielefeld University,  
P.O. Box 100131, 33501 Bielefeld, Germany
} \email{akemann@physik.uni-bielefeld.de}

\author[T. Nagao]{Taro Nagao}
\address{Graduate School of Mathematics, Nagoya University, Chikusa-ku, Nagoya 464-8602, Japan}\email{nagao@math.nagoya-u.ac.jp}

\author[I. Parra]{Iv\'an Parra}
\address{Faculty of Physics,  Bielefeld University,  
P.O. Box 100131, 33501 Bielefeld, Germany
} \email{iparra@physik.uni-bielefeld.de}

\author[G. Vernizzi]{Graziano Vernizzi}
\address{Siena College, 515 Loudon Road
Loudonville, NY 12211, USA}\email{gvernizzi@siena.edu}

\commby{}
\begin{abstract}

We show that several families of classical orthogonal polynomials on the real line are also orthogonal 
on the interior of an ellipse in the complex plane, subject to a weighted planar Lebesgue  measure. 
In particular these include Gegenbauer polynomials 
$C_n^{(1+\alpha)}(z)$ for $\alpha>-1$ containing the Legendre polynomials $P_n(z)$, and the subset $P_n^{(\alpha+\frac12,\pm\frac12)}(z)$  of the Jacobi polynomials.
These polynomials  provide an orthonormal basis and the corresponding weighted Bergman space forms a complete metric space.  This leads to a certain family of Selberg integrals in the complex plane.
We recover the known orthogonality of Chebyshev polynomials of first up to fourth kind. The limit $\alpha\to\infty$  leads back to the known Hermite polynomials orthogonal in the entire complex plane. When the ellipse degenerates to a circle we obtain the weight function and monomials known from the determinantal point process of the ensemble of truncated unitary random matrices.

\end{abstract}

\maketitle
\tableofcontents
\section{Introduction}\label{SectionIntroduction}

Orthogonal polynomials in the complex plane play an important role for non-Hermitian random matrix theory. A prominent example is the elliptic Ginibre ensemble  with complex normal matrix elements, having different variances for their real and imaginary parts,  \cite{SCSS}. Its complex eigenvalues follow a determinantal point process, with its kernel constituted by the Hermite polynomials orthogonal in the complex plane \cite{FKS98}. Likewise, the chiral partner of this ensemble leads to a kernel of generalised Laguerre polynomials orthogonal in the complex plane \cite{Osborn,ABe}, cf. \cite{Karp}. 
The respective kernels allow for a complete characterisation of all complex eigenvalue correlation functions of these ensembles of random matrices. Moreover, in the limit of weak non-Hermiticity introduced in \cite{FKS},
these nontrivial polynomials allow to study an interpolation between the statistics of real eigenvalues of Hermitian random matrices on the one hand, e.g. of the Gaussian Unitary Ensemble characterised by Hermite polynomials on the real line,  and those of complex eigenvalues e.g. of the Ginibre ensemble, being characterised by monomial polynomials in the complex plane. We refer to \cite{APh} to a list of interpolating limiting kernels known to date.

In this paper we ask the question whether further classical orthogonal polynomials on the real line also form a set of orthogonal polynomials on a two dimensional domain in the complex plane. Orthogonal polynomials on the real line or subsets thereof, as well as those on one-dimensional curves on the complex plane - typically the unit circle - are a classical topic in Mathematics \cite{Szego}. Therefore, it is quite surprising that relatively few works have addressed this question. 
The orthogonality of Chebyshev polynomials of the second kind on the interior of an unweighted ellipse probably goes back to \cite{Henrici}.
The fact that Hermite polynomials are also orthogonal with respect to a Gaussian weight in the complex plane was first shown in 1990 \cite{EM}, cf. \cite{PdF} for an independent proof. Generalised Laguerre polynomials in the complex plane were found in the context of applications to quantum field theory in \cite{Osborn}, see \cite{ABe} for a concise orthogonality proof. The orthogonality of all Chebyshev polynomials of first to fourth kind on an ellipse can be found in \cite{Mason}.

While the Gram-Schmidt construction of orthogonal polynomials on any subset of the real line and in  the complex plane is completely analogous, given that all moments exist, cf. \cite{Walter}, the fact that the former always satisfy a three-step recursion relation is special. 
While Lempert \cite{Lempert} showed that we cannot expect any finite term recurrence for orthogonal polynomials in the complex plane in general, it was shown much more recently that the existence of a finite term recurrence relation on a unweighted bounded domain with sufficiently regular boundary implies that the domain is an ellipse and the recursion depth is three \cite{PS,KS}. 
This limits our search to elliptic domains as our polynomials originating from the real line do have a three step recurrence. We note, however, that the aforementioned results  \cite{PS,KS} only apply to unweighted domains. 
For the Chebyshev polynomials of first, third and fourth kind, the weight function on the ellipse is no longer flat \cite{Mason}.

In this work we obtain the following results. We show that the classical Gegenbauer or ultraspherical polynomials provide a family of planar orthogonal polynomials on the interior of a weighted ellipse. They generalise the monomials that appear in the determinantal point process on the unit disc, obtained from the ensemble of truncated unitary  random matrices \cite{ZS}. Furthermore, we find a subset 
of the Jacobi polynomials to be orthogonal on a weighted ellipse. These findings allow to recover the orthogonality of all four Chebyshev polynomials from \cite{Mason}. 
All these planar orthogonal polynomials  lead to examples for Selberg- (or Mehta-) type integrals in the complex plane containing a Vandermonde determinant modulus squared, when determining the normalisation of the corresponding determinantal point processes, see \cite{Peter-Selberg} for a review on Selberg integrals.

As an application we use the Gegenbauer polynomials to construct further (non-classical) planar orthogonal polynomials on a weighted ellipse, that do not satisfy a recursion relation of finite depth. In that sense the ellipse is not a special domain in the complex plane, once nontrivial weight functions are allowed. 
At present we do not know a random matrix model that leads to a determinantal point process with a kernel of Gegenbauer or a subset of Jacobi polynomials - apart from a trivial normal matrix representation. In this work we restrict ourselves to polynomials of finite degree. The asymptotic of the Bergman kernel in the limit of weak non-Hermiticity, both in 
 in the bulk and at the edge of the ellipse, will be presented elsewhere \cite{AKNP}.

The remainder of this article is organised as follows. 
To prepare the ground in Section \ref{Complete metric spaces} we show that the weighted Bergman space on the ellipse is a complete metric space. To that aim in  Section 
\ref{Gegenbauer polynomials} 
we prove the orthogonality of the  Gegenbauer polynomials 
$C_n^{(1+\alpha)}(z)$ of even degree, for $\alpha>-1$, 
with respect to the inner product on the weighted ellipse.
The case with an odd degree is very similar and presented in Appendix \ref{La-odd}. 
This immediately implies the orthogonality of Legendre polynomials $P_n(z)$ as well, and we recover the orthogonality of Chebyshev polynomials of the second kind $U_n(z)$.
In Appendix \ref{La-proof2} an alternative orthogonality proof for Gegenbauer polynomials independent of the degree is given, that in contrast relies on the known orthogonality of the Chebyshev polynomials of the second kind on the unweighted ellipse. The proof for the latter from \cite{Hormander} is collected in Appendix \ref{app:Uproof} for completeness. 
In Section \ref{sec:Jacobi} we prove that two families of particular Jacobi polynomials $P_n^{(\alpha+\frac12,\pm\frac12)}(z)$, for $\alpha>-1$,   are orthogonal  on weighted ellipses.
 The known orthogonality of the Chebyshev polynomials of third, fourth and first kind $V_n(z)$, $W_n(z)$ and $T_n(z)$ respectively follow as a consequence. 
In Section \ref{Bergman-recursion} we construct an explicit example for orthogonal polynomials on a weighted ellipse, that do not satisfy a recursion relation of finite depth. 
The construction is based on Gegenbauer polynomials and  the Heine-formula for planar  orthogonal polynomials. 
Here, we also present the Selberg integral based on the family of Gegenbauer polynomials $C_n^{(1+\alpha)}(z)$ as an example, that can be analytically continued in  $\alpha$.

\ \\
 \textbf{Acknowledgements.} 
We acknowledge support by the
German research council (DFG) through CRC 1283:
“Taming uncertainty and profiting from randomness and
low regularity in analysis, stochastics and their applications"
(GA), by the Japan Society for the Promotion of Science, KAKENHI 25400397 (TN), 
 and by the grant  DAAD-CONICYT/Becas Chile, 2016/91609937 (IP). 
 The hospitality of the MFO Oberwolfach during the workshop on “Free Probability" is also acknowledged (GA), in particular their excellent library service.

 \section{Weighted Bergman space on the interior of an ellipse}
 \label{Complete metric spaces}
 
To begin let us fix some notation for the quantities to be considered. For  $a>b>0$ the following provides an explicit  parametrisation of the interiour of an ellipse $E$:  
 \begin{align}
  \label{ellipsedef}
  & E\,=\, \{z \in \mathds{C}: h(z):=(\Re{z})^2/a^2+(\Im{z})^2/b^2<1 \} \ .
\end{align}
For  $ 0 < p < \infty$ and $-1 < \alpha < \infty$, we will denote by $A_{\alpha}^{p}:=A_{\alpha}^{p}(E)\subseteq L^{p}(E,\d{A}_{\alpha})$  the (weighted) Bergman space of the ellipse $E$, i.e. the subspace of analytic functions in $L^{p}(E,\d{A}_{\alpha})$ with finite $p$-norm. The area measure  
\begin{align}
  \label{ellipseweight}
  & \d{A}_{\alpha}(z)\,=\,(1+\alpha)(1-h(z))^{\alpha}\d{A}(z)
\end{align} 
is defined in terms of the normalised area measure on the ellipse $\d{A}(z)=\d x\d y/(\pi ab)$, with $z=x+iy$, 
together with $h(z)$ defined in the parametrisation of the ellipse \eqref{ellipsedef}. 
 It is not difficult to see that it is normalised $\forall\alpha>-1$:
 \begin{equation}
 \label{normAa}
 \int_{E}\d{A}_{\alpha}(z)\,=\,
\frac{1+\alpha}{\pi}\int_0^1\d rr\int_0^{2\pi}\d\theta\, (1-r^2)^{\alpha}
\,=\,
 1\ ,
 \end{equation}
after changing variables to
\begin{equation}
\label{cambio1}
 x=ar\cos(\theta)\ ,\ \ y=br\sin(\theta)\ ,
\end{equation}
with $r \in [0,1),\ \theta \in [0, 2\pi]$ and Jacobian $J(r,\theta)=abr$.
 For $1\leq p<\infty$ the associated $L^p$-norm is definded  by
\begin{align}
  \label{ine}
||f||_{p,\alpha}=\le\int_{E}|f(z)|^{p} \,\d A_{\alpha}{(z)}\ri^{1/p},
\end{align}
and for $0<p<1$ the corresponding metric is given by 
\be
d(f,g)\ =\ ||f-g||_{p,\alpha}^{p} \ .
\ee 

In this section we show that  the Bergman space $A_{\alpha}^{p}$ is a Banach space when $1\leq p <\infty$, and a complete metric space when $0<p<1$. The proof is quite standard and follows the lines of Corollary 1.12 and Proposition 1.13 in \cite{Conway}.
 
Let $t \in E$ and $0<\rho<\mbox{dist}(t,\partial E)=:d$ be arbitrary. We define the smaller ellipse
\begin{align}
  \label{ellipser}
  & E_\rho\,=\, \{z \in \mathds{C}: h_\rho(z):=(\Re{z})^2/(a-\rho/2)^2+(\Im{z})^2/(b-\rho/2)^2\leq1 \} \ ,
\end{align}
 and suppose that there is a point $z_0 \in B(t,\rho/2)\setminus E_\rho $:
\begin{align}
  \label{ine2}
  \{|z_0-w|\,: w \in \partial E\} \subseteq \{|z-w|\,: z \in B(t,\rho/2); w \in \partial E\}\ .
\end{align}
Taking the infimum on both sides of (\ref{ine2}), we obtain 
\begin{align}
  \label{inee}
\mbox{dist}(B(t,\rho/2),\partial E) \leq \mbox{dist}(z_0, \partial E)\ .
\end{align}
But (\ref{inee}) implies that $d-\rho/2\leq \rho/2$, therefore $B(t,\rho/2)\subseteq E_\rho$. In consequence we obtain
\begin{align}
  \label{sup}
\sup_{z\in B(t,\rho/2)} h(z) \leq \sup_{z\in E_{\rho}} h(z) \leq h(z_{*})=:c(\rho) \  ,\quad z_{*} \in \partial E_{\rho}\ .
\end{align}
It is easy to see that $0<c(\rho)<1$,  and it can be computed explicitly by introducing a  Lagrange multiplier, for example.

Thus, given $f \in A_{\alpha}^{p}$, $B(t,r)\subseteq E$ with positive minimum distance to the boundary $\partial E$, i.e. $0<\rho<\mbox{dist}(B(t,r),\partial E)$, we can find another positive constant $C>0$ such that
\begin{eqnarray}
|f(z)|^{p}&\leq& \frac{4}{\pi \rho^2} \int_{B(z,\rho/2)}|f(w)|^p \,\d A{(w)} 
\no\\
&\leq& C\int_{B(z,\rho/2)}|f(w)|^p \,\d A_{\alpha}{(w)}\no\\
&\leq& C\int_{E}|f(w)|^p \,\d A_{\alpha}{(w)}\no\\
&=& C\|f\|_{p,\alpha}^{p}\;\;\;\;\;\;\;\;\text{for}\; z \in B(t,r)\ .
 \label{prop}
\end{eqnarray} 
 In the first step we have used the subharmonicity of $|f|^p$. In the second step the upper bound is trivial for negative $-1<\alpha<0$, due to $0\leq h(z)$, whereas for positive $\alpha>0$ we have used the estimate from \eqref{sup}.
The  statement \eqref{prop} can be summarised in the following 
\begin{prop}\label{bound}
Let $ 0 < p < \infty$ and $-1 < \alpha < \infty$, and $K$ be a compact subset of $E$, with positive minimum distance to $\partial E$. Then, there is a positive constant C such that  
\begin{align}
  \label{sup}
 \sup_{K} |f(z)|^p \leq C\|f\|_{p,\alpha}^p\ , \no
\end{align}
for all $f \in A_{\alpha}^{p}$.
\end{prop}
One immediate consequence of this proposition  is that any Cauchy sequence $\{f_n\} \in A_\alpha^p$ is locally bounded, and so by Montel's Theorem it constitutes a normal family. Thus, some
subsequence converges locally uniformly in $E$, to a function in $A_\alpha^p$, and we have 
\begin{cor}\label{cor:closed}
For every $0 < p < \infty$,  $-1 < \alpha <\infty$,  the weighted Bergman space $A_{\alpha}^{p}$ is closed in $L^p(E, dA_{\alpha})$. 
\end{cor}
\begin{proof}
Let $\{f_n\}$ be a Cauchy sequence in $A_{\alpha}^{p}$ and $f \in L^p(E, dA_{\alpha})$, such that  it holds $\int |f_n-f|^p \d A_{\alpha}\rightarrow 0$ as $n \rightarrow \infty$.
By Montel's Theorem $\{f_n\}$ converges locally uniformly to a function $g$ that is analytic in $E$. Since
$\| f_n - f \|_{p,\alpha}^{p} \rightarrow 0$, this implies  that $f_n$ converges in measure to $f$. By  Riesz' Theorem there is a subsequence  $\{f_{n_{k}}\}$ such
that $f_{n_{k}}(z) \rightarrow f(z)$ a.e. Thus $f = g$ a.e.,  and so $f \in A_{\alpha}^{p}$.
\end{proof}
For $p\geq1$ it follows from Corollary \ref{cor:closed} that the Bergman space is a Banach space, and in particular for $p=2$ a Hilbert space.
In the next section we will consider the Bergman space for $p=2$, $A_{\alpha}^{2}$ as a Hilbert space, with the 
notion for the inner product defined as 
\begin{equation}
\label{innerproduct}
\langle f,g\rangle_\alpha := \Eint f(z)\overline{g(z)}\,\d A_\alpha{(z)} \ ,
\end{equation}
for two integrable functions $f,g\in A_\alpha^2$.

For the analyticity it is of course important that we consider the interior of the ellipse $E$ \eqref{ellipsedef}, being an open set. 
Because the boundary of the ellipse $\partial E$ is one-dimensional and of measure zero in the complex plane, all integrals over $\overline{E}=E\cup\partial E$ agree,
\begin{equation}
\label{intbarE}
\Eint f(z)\overline{g(z)}\,\d A_\alpha{(z)} = \int_{\overline{E}} f(z)\overline{g(z)}\,\d A_\alpha{(z)} \ ,
\end{equation}
for $-1<\alpha$.
We will come back to this point when relating to the ensemble of truncated unitary matrices in Remark \ref{rem:tUnitary}.


 \section{Orthogonality of Gegenbauer and Legendre polynomials}\label{Gegenbauer polynomials}

For any non-negative integer $n$ let us define the polynomials 
\begin{eqnarray} \label{bin}
p_{n}^{(\alpha)}(z)&:=& \frac{1}{\sqrt{h_n}} C_{n}^{(1+\alpha)}\le\frac{z}{c}\ri\ ,
\end{eqnarray}
where $C_{n}^{(1+\alpha)}(x)$ are the standard Gegenbauer polynomials on the real line having real coefficients, now taken with a complex argument. We recall that the ellipse $E$ in \eqref{ellipsedef} defining the inner product \eqref{innerproduct} is parametrised by the real numbers $a>b>0$.  The constant $c=\sqrt{a^2-b^2}>0$ is then the right focus of the 
ellipse $E$, and we define by
\begin{eqnarray} \label{binnorm}
h_{n}&:=&h_{n}(a,b)= \frac{1+\alpha}{1+\alpha+n}C_n^{(1+\alpha)}\le\frac{a^2+b^2}{a^2-b^2}\ri >0 \ ,
\end{eqnarray} 
the norms of the Gegenbauer polynomials in the complex plane. Their positivity follows from \eqref{C2n}, and a short argument goes as follows.  Because the Gegenbauer polynomials \eqref{bin} have all their zeros in $(-1,1)$, the fact that the argument $(a^2+b^2)/(a^2-b^2)>1$ in \eqref{binnorm} is to the right of this interval, together with the positivity of the leading coefficient of the $C_{n}^{(1+\alpha)}(x)$ there \cite{NIST},  
leads to the positivity of $h_n$ for all integers $n\geq0$. 
We claim the following
\begin{thm}
\label{propONB}
The set of polynomials $\{p_n^{(\alpha)}\}_{n\in\mathbb{N}}$ defined  in \eqref{bin} forms a orthonormal basis for $A_{\alpha}^{2}$. 
\end{thm}
In view of the previous subsection we need to prove the orthonormality and completeness of the basis. The former is shown in the following lemma, whereas the completeness is deferred to the very end of this section. 
\begin{lem}
\label{COP-lemma}
For the sequence of Gegenbauer polynomials $\{C_{n}^{(1+\alpha)}\}_{n\in\mathbb{N}}$, with $-1<\alpha$, on the domain \eqref{ellipsedef} with weight \eqref{ellipseweight}, the following orthogonality relation holds 
 \begin{eqnarray} \label{OP}
\Eint C_{m}^{(1+\alpha)}\left(\frac{z}{c}\right)C_{n}^{(1+\alpha)}\left(\frac{\overline{z}}{c}\right)\,\d A_\alpha{(z)}&=& \frac{1+\alpha}{1+\alpha+n}C_n^{(1+\alpha)}\le\frac{a^2+b^2}{a^2-b^2}\ri\delta_{nm}\ ,
\end{eqnarray}
where $a>b>0$ and $c=\sqrt{a^2-b^2}$.
\end{lem}
\begin{rem}
In this section we will present an elementary proof of the orthogonality relation \eqref{OP}. Due to the reflection symmetry of the weight function, domain and parity of the polynomials, the proof can be split into even and odd polynomials separately. Because these two cases are very similar we only present the one for the even polynomials in the main body of the paper here. For completeness we have put the proof for the odd polynomials into the Appendix \ref{La-odd}.  A second independent proof valid for the orthogonality of the even and odd polynomials alike is presented in Appendix \ref{La-proof2}. 
It starts by assuming the known orthogonality of Chebyshev polynomials  of the second kind $U_n$ on the unweighted ellipse, which can be found in  \cite{Henrici} and that is reproduced for completeness in Appendix \ref{app:Uproof}. While the orthogonality proof in Appendix \ref{La-proof2} is more elegant, there the determination of the norms $h_n$ is much more cumbersome and therefore will not be presented.

In addition to the proof presented below, the orthogonality of the Chebyshev polynomials of the first up to forth kind follows as a corollary, as demonstrated in Section \ref{sec:Jacobi}.  This establishes an independent proof of \cite{Mason}, where the orthogonality of all four kind was shown previously.
\end{rem}

\begin{proof} It is sufficient to show that for all $m\in\mathbb{N}$
\begin{equation} \label{pd}
\Eint C_{m}^{(1+\alpha)}\left(\frac{z}{c}\right)\left(\frac{\overline{z}}{c}\right)^{j}\,\d A_\alpha{(z)}= 0 \;,\;\;\text{for}\; j=0,1,...,m-1\ ,
\end{equation}
holds true.
Since both the weight function \eqref{ellipseweight} and domain \eqref{ellipsedef} are invariant under the reflexion $z\to-z$,  and the polynomials have parity, 
$C_{n}^{(1+\alpha)}(-z)=(-1)^nC_{n}^{(1+\alpha)}(z)$, without restriction we assume that either $m = 2n$ and $j = 2l$ are both even, 
or $m = 2n+1$ and $j = 2l+1$ are both odd, and $l<n$.
In the following we will only present the even-even case.
The odd-odd case follows from the same line of arguments and is collected in Appendix \ref{La-odd} for the reader's convenience.

We rewrite the integral \eqref{pd} with $z=x+iy$ in terms of elliptic coordinates. To that aim we change variables as follows:
\begin{equation}
\label{cambio}
 x=ar\cos(\theta)\ ,\ \ y=br\sin(\theta)\ ,\ \ \mbox{with}\ \  r \in [0,1),\; \theta \in [0, 2\pi]\ .
\end{equation}
The Jacobian for this transformation reads $J(r,\theta)=abr$, and we obtain for the complex arguments 
 \begin{eqnarray} \label{change}
\frac{z(r,\theta)}{c}&=& \frac{r}{2}( Re^{i\theta}+ R^{-1}e^{-i\theta})\ ,\ \;\text{with}\; R:=\frac{a+b}{c}=\sqrt{\frac{a+b}{a-b}}\ .
\end{eqnarray}
We also obtain $h(z)=r^2$ from \eqref{ellipsedef}. This leads to the following expression
\begin{eqnarray} 
&&\Eint C_{2n}^{(1+\alpha)}\left(\frac{z}{c}\right)\left(\frac{\overline{z}}{c}\right)^{2l}\,\d A_\alpha{(z)}=
\nonumber\\ 
&&=\frac{1+\alpha}{\pi}\int_0^1\d rr\int_0^{2\pi}\d\theta\, C_{2n}^{(1+\alpha)}\left(\frac{z(r,\theta)}{c}\right)\left(\frac{\overline{z(r,\theta)}}{c}\right)^{2l}\,(1-r^2)^{\alpha}.
\label{pd1}
\end{eqnarray}
The even Gegenbauer polynomials can be written in terms of Gau{\ss}' hypergeometric function in the following way, see e.g. \cite[8.932.2]{Grad}
\begin{eqnarray} 
&&C_{2n}^{(1+\alpha)}\left(\frac{z(r,\theta)}{c}\right)= 
\frac{(-1)^n\Gamma(n+1+\alpha)}{\Gamma(n+1)\Gamma(1+\alpha)}
F\left(-n, n+\alpha+1;\frac12; \frac{z(r,\theta)^2}{c^2}\right)
\no\\
&&=\frac{(-1)^n}{2\G(1+\alpha)n!}\sum_{p=0}^{n}\sum_{k=0}^{2p}(-1)^p{n \choose p}{2p \choose k}\frac{\G(1+\alpha +n+p)\G(p)}{\G(2p)}
r^{2p} R^{2(k-p)} e^{2i\theta(k-p)}
\no\\
&&=\frac{(-1)^n}{2\G(1+\alpha)n!}\sum_{p=0}^{n}\sum_{k=0}^{2p}(-1)^p{n \choose p}{2p \choose k}\frac{\G(1+\alpha +n+p)\G(p)}{\G(2p)}
r^{2p} R^{2(p-k)} e^{2i\theta(p-k)}.
\no\\
\label{eq:par}
\end{eqnarray}
Here, we introduced two representations  to be both used below, using the binomial theorem for \eqref{change} in two equivalent ways.
In order to prepare the integration in \eqref{pd1}, we spell out the complex conjugated variable to the power $2l$:
\begin{eqnarray} \label{bin1}
\left(\frac{\overline{z(r,\theta)}}{c}\right)^{2l}&=& 
\le \frac{r}{2}\ri^{2l}  \le Re^{-i\theta}+ R^{-1}e^{i\theta}\ri^{2l}
\no\\
&=& \le \frac{r}{2}\ri^{2l} 
\left[\sum_{k=1}^{l}{2l \choose k+l}R^{-2k}e^{2i\theta k}+{2l \choose l}+\sum_{k=1}^{l}{2l \choose k+l}R^{2k}e^{-2i\theta k}\right].
\end{eqnarray}
From the radial integral in \eqref{pd1} we obtain, including all prefactors
\begin{equation}
\label{radial}
\frac{1+\alpha}{\pi}\int_0^1\d r\,r^{2p+1} \frac{r^{2l}}{2^{2l}} (1-r^2)^{\alpha}=\frac{\Gamma(2+\alpha)\Gamma(1+p+l)}{2^{2l+1}\pi\Gamma(2+\alpha+p+l)}.
\end{equation}
For the remaining angular integration we thus have
\begin{eqnarray}
&&\Eint C_{2n}^{(1+\alpha)}\left(\frac{z}{c}\right)\left(\frac{\overline{z}}{c}\right)^{2l}\,\d A_\alpha{(z)}=
\nonumber\\ 
&&= \frac{(1+\alpha)(-1)^n}{2^{2l+1}\pi}\sum_{k'=1}^{l}\sum_{p=0}^{n}\sum_{k=0}^{2p}{2l \choose k'+l}\frac{(-1)^p\G(1+\alpha +n+p)\G(1+l+p)}{(n-p)!k!(2p-k)!\G(2+\alpha+l+p)}
\no\\
&&\qquad \qquad\qquad\qquad\times  
R^{2(p-k-k')} \int_{0}^{2\pi}\d{\theta}\,e^{2i\theta(p-k+k')}
\no\\
&&\quad +\frac{(1+\alpha)(-1)^n}{2^{2l+1}\pi}\sum_{p=0}^{n}\sum_{k=0}^{2p}{2l \choose l}
\frac{(-1)^p\G(1+\alpha +n+p)\G(1+l+p)}{(n-p)!k!(2p-k)!\G(2+\alpha+l+p)}\no\\
&&\qquad \qquad\qquad\qquad\times  
R^{2(p-k)} \int_{0}^{2\pi}\d{\theta}\,e^{2i\theta(p-k)}
\no\\
&&\quad+  \frac{(1+\alpha)(-1)^n}{2^{2l+1}\pi}\sum_{k'=1}^{l}\sum_{p=0}^{n}\sum_{k=0}^{2p}{2l \choose k'+l}\frac{(-1)^p\G(1+\alpha +n+p)\G(1+l+p)}{(n-p)!k!(2p-k)!\G(2+\alpha+l+p)}
\no\\
&&\qquad \qquad\qquad\qquad\times  R^{2(k-p+k')} \int_{0}^{2\pi}\d{\theta}\,e^{2i\theta(k-p-k')}.
\no\\
\label{thetaint}
\end{eqnarray}
In the first step we have already simplified the binomial factors and Gamma-functions from \eqref{eq:par}. 
Notice that in the first two terms, obtained from integrating over the first two contributions on the right-hand side of \eqref{bin1}, we have used the second identity in \eqref{eq:par}, whereas for the last sum from \eqref{bin1} we have used the first form of identity in \eqref{eq:par}.
We now evaluate each of the multiple sums in \eqref{thetaint} individually. In the last triple sum we have $k=p+k'$ due to the angular integration, and because of $k\leq 2p$ and thus $k'\leq p$ we obtain for it
\begin{eqnarray}
&&\frac{(1+\alpha)(-1)^n}{2^{2l}}\sum_{k'=1}^{l}{2l \choose k'+l}
\sum_{p=k'}^{n}\frac{(-1)^p\G(1+\alpha +n+p)\G(1+l+p)}{(n-p)!(k'+p)!(p-k')!\G(2+\alpha+l+p)}R^{4k'}
\no\\
&&:=\frac{(1+\alpha)(-1)^n}{2^{2l}}\sum_{k'=1}^{l}{2l \choose k'+l}a_{k'}R^{4k'}.
\label{thirdsum}
\end{eqnarray}
It is a polynomials in $R$ of degree $4l$. We have to show that all its coefficients $a_{k'}=a_{k'}(n,l)$ vanish for $l<n$. Before we do that let us compute the other sums in \eqref{thetaint}.
From the second term in \eqref{thetaint}, the double sum, we obtain from $p=k$
\begin{equation}
\label{secondsum}
\frac{(1+\alpha)(-1)^n}{2^{2l}}{2l \choose l}\sum_{p=0}^{n}\frac{(-1)^p\G(1+\alpha +n+p)\G(1+l+p)}{(n-p)!(p!)^2\G(2+\alpha+l+p)}
=\frac{(1+\alpha)(-1)^n}{2^{2l}}{2l \choose l}a_{0}\ ,
\end{equation}
which is $R$-independent. It the same as the contribution in \eqref{thirdsum} for $k'=0$. For the first triple sum in \eqref{thetaint} we have again $k=p+k'$ and and thus $k'\leq p$:
\begin{equation}
\label{firstsum}
\frac{(1+\alpha)(-1)^n}{2^{2l}}\sum_{k'=1}^{l}{2l \choose k'+l}
\sum_{p=k'}^{n}\frac{(-1)^p\G(1+\alpha +n+p)\G(1+l+p)}{(n-p)!(k'+p)!(p-k')!\G(2+\alpha+l+p)}R^{-4k'}.
\end{equation}
It agrees with \eqref{thirdsum} replacing $R\to R^{-1}$. So in summary if we can show that all coefficients $a_{k'}$ vanish for $k'=0,1,\ldots,l$ when $l<n$ we are done. 
This can be seen as follows. From the definition \eqref{thirdsum} we have, after a change of variables,
\begin{eqnarray}
a_{k}&=&
\frac{(-1)^k}{(n-k)!}\,\sum_{p=0}^{n-k}(-1)^p{n-k \choose p}\frac{\G(1+\alpha +n+k+p)\G(1+l+k+p)}{(2k+p)!\G(2+\alpha+l+k+p)}
\no\\
&=&\frac{(-1)^k}{(n-k)!\G(1+\alpha)}\,\int_{0}^{1}\,\d{x}\, x^{l+k}(1-x)^{\alpha}\sum_{p=0}^{n-k}(-1)^p{n-k \choose p}\frac{\G(1+\alpha +n+k+p)}{\G(1+2k+p)}x^p
\no\\
&=&\frac{(-1)^k\G(1+\alpha+n+k)}{(n-k)!\G(1+\alpha)(2k)!}\,\int_{0}^{1}\,\d{x}\, x^{l+k}(1-x)^{\alpha}F(-n+k, 1+\alpha+n+k; 1+2k; x).
\no\\
\label{ankint}
\end{eqnarray}
This reduces the problem to show that the integral containing the hypergeometric function vanishes, when $l <n$ and $ \alpha> -1$.
Let us introduce a regularising parameter $\varepsilon>0$. We then have 
$$\displaystyle \left| x^{l+k+\varepsilon}(1-x)^{\alpha}F(-n+k, 1+\alpha+n+k, 1+2k, x) \right| \leq C_{F}\;x^{l+k}(1-x)^{\alpha}\;,\;\;\; x \in [0,1]\ ,$$
for some constant $C_F$. Since $\displaystyle x^{l+k}(1-x)^{\alpha} \in L^{1}([0,1]) $, by Lebesgue's dominated convergence theorem, we have 
\begin{equation}\label{prueba4}
\begin{split}
&\int_{0}^{1}\,\d{x}\, x^{l+k}(1-x)^{\alpha}F(-(n-k), 1+\alpha+n+k; 1+2k; x)\\
&=\lim_{ \varepsilon \to 0}\int_{0}^{1}\,\d{x}\, x^{l+k+\varepsilon}(1-x)^{\alpha}F(-(n-k), 1+\alpha+n+k; 1+2k; x)\\
&=\lim_{\varepsilon \to 0} \frac{\G(1+2k)\G(1+l+k+\varepsilon) \G(1+\alpha+n-k)\G(n-l-\varepsilon)}{\G(1+k+n)\G(2+\alpha+n+l+\varepsilon)\G(-\varepsilon-(l-k))}\\
&=\lim_{\varepsilon \to 0} \frac{(-1)^{l-k-1}\G(1+2k)\G(1+l+k+\varepsilon) \G(1+\alpha+n-k)\G(l+1-k+\varepsilon)}{\pi\G(1+k+n)\G(2+\alpha+n+l+\varepsilon)}\\
&\qquad \qquad \times \G(n-l-\varepsilon)\sin(\pi\varepsilon)\ .
\end{split}
\end{equation}
In the second step we have used  the following integral, see \cite[7.512.2]{Grad}
\begin{eqnarray*}
\int_{0}^{1}\,t^{\rho-1}(1-t)^{\beta-\gamma-m}\, F(-m, \beta; \gamma; t)\, \d{t}=\frac{\G(\gamma)\,\G(\rho)\,\G(\beta-\gamma+1)\G(\gamma-\rho+m)}{\G(\gamma+m)\,\G(\beta-\gamma+\rho+1)\G(\gamma-\rho)}\no\\
\mbox{for}\quad
 m=0, 1, 2,...;  \Re  \rho > 0, \Re(\beta-\gamma)>m-1\ ,
\end{eqnarray*}
and in the next step Euler's reflection formula. Finally the limit
\begin{equation}
\label{sinepsilon}
\lim_{\varepsilon\to0} \Gamma(n-l-\varepsilon)\sin(\pi\varepsilon)=\left\{
\begin{array}{cc}
-\pi & l=n\\
0 & l<n\\
\end{array}
\right.
\end{equation} 
establishes the claimed orthogonality of \eqref{pd} for even indices.

In order to compute the squared norm on the right-hand side of \eqref{OP}, we first compute \eqref{thetaint} for $n=l$. For that purpose we summarise the results for the coefficients $a_k$ in \eqref{ankint} that follows from the result above:
\begin{equation}
a_k(n,n)=\frac{(-1)^n\Gamma(1+\alpha+n+k)\Gamma(1+\alpha+n-k)}{\Gamma(1+\alpha)\Gamma(2n+\alpha+2)}\ ,
\label{ank}
\end{equation}
to be inserted into \eqref{thirdsum}. We thus obtain for this, as well as for \eqref{firstsum} at $n=l$,
\begin{eqnarray}
&&\sum_{k'=1}^{n}{2n \choose k'+n}a_{k'}(n,n)R^{\pm4k'}\no\\
&=& \frac{(-1)^n(2n)!}{\Gamma(1+\alpha)\Gamma(2n+\alpha+2)}\sum_{k'=1}^{n}\frac{\Gamma(1+\alpha+n+k')\Gamma(1+\alpha+n-k')}{(n-k')!(n+k')!}R^{\pm4k'}
\no\\
&=& \frac{(-1)^n(2n)!}{\Gamma(1+\alpha)\Gamma(2n+\alpha+2)}\sum_{k=n+1}^{2n}\frac{\Gamma(1+\alpha+k)\Gamma(1+\alpha+2n-k)}{(2n-k)!k!}R^{\mp4(n-k)}
\no\\
&=& \frac{(-1)^n(2n)!}{\Gamma(1+\alpha)\Gamma(2n+\alpha+2)}\sum_{k=0}^{n-1}\frac{\Gamma(1+\alpha+2n-k)\Gamma(1+\alpha+k)}{k!(2n-k)!}R^{\pm4(n-k)}\ ,
\label{anksum}
\end{eqnarray}
after relabelling the sum twice. With this result it is easy to see that we can write the three contributions \eqref{thirdsum}, \eqref{secondsum} and \eqref{firstsum} at $n=l$ to \eqref{thetaint} as a single sum as
\begin{eqnarray}
&&\Eint C_{2n}^{(1+\alpha)}\left(\frac{z}{c}\right)\left(\frac{\overline{z}}{c}\right)^{2l}\,\d A_\alpha{(z)}=\no\\
&&=\delta_{n,l}\frac{(1+\alpha)(2n)!}{2^{2n}\Gamma(1+\alpha)\Gamma(2n+\alpha+2)}\sum_{k=0}^{2n}\frac{\Gamma(1+\alpha+k)\Gamma(1+\alpha+2n-k)}{\Gamma(2n-k+1)\Gamma(k+1)}R^{4(n-k)}.
\label{C2n}
\end{eqnarray}
 The remaining sum can be related to a single Gegenbauer polynomial as follows. Because this sum is invariant under $k\to2n-k$ we can write it as
 \begin{eqnarray}
&=&\frac12 \sum_{k=0}^{2n}\frac{\Gamma(1+\alpha+k)\Gamma(1+\alpha+2n-k)}{\Gamma(2n-k+1)\Gamma(k+1)}\left(R^{4(n-k)}+R^{-4(n-k)}\right)\no\\
 &=&\sum_{k=0}^{2n}\frac{\Gamma(1+\alpha+k)\Gamma(1+\alpha+2n-k)}{\Gamma(2n-k+1)\Gamma(k+1)}\cosh[(2n-2k)\ln(R^2)]\no\\
 &=&\Gamma(1+\alpha)^2 C_{2n}^{(1+\alpha)}\left(\frac{a^2+b^2}{a^2-b^2}\right)\ .
 \label{Cab}
 \end{eqnarray}
 In the last step we have used the (analytically continued) relation \cite[18.5.11]{NIST}
\begin{equation}
\label{CcosId}
C_{j}^{(1+\alpha)}(\cos\theta)=\sum_{l=0}^j\frac{(1+\alpha)_l(1+\alpha)_{j-l}}{l!(j-l)!}\cos((j-2l)\theta)\ ,
\end{equation} 
with $(a)_n=\Gamma(a+n)/\Gamma(n)$ being the Pochhammer symbol, together with
 \begin{equation}
 \label{Rab}
 \cosh[\ln(R^2)]=\frac12(R^2+R^{-2})=\frac{a^2+b^2}{a^2-b^2}\ ,
 \end{equation}
that follows from \eqref{change}. In order to obtain \eqref{OP} we still need to multiply \eqref{C2n} with the leading power of the Gegenbauer polynomial which is easy to obtain from the first line of \eqref{eq:par}, cf. \cite{NIST}
\begin{equation}
C_{2l}^{(1+\alpha)}(x)=\frac{\Gamma(2l+1+\alpha)2^{2l}}{\Gamma(1+\alpha)(2l)!} x^{2l} + O(x^{2l-2})\ .
\end{equation}
Because the lower powers give zero, combined with \eqref{Cab} we finally have 
\begin{eqnarray}
&&\ \Eint C_{2n}^{(1+\alpha)}\left(\frac{z}{c}\right)C_{2l}^{(1+\alpha)}\left(\frac{\overline{z}}{c}\right)\d A_\alpha{(z)}
=\delta_{2n,2l}
\frac{(1+\alpha)}{(2n+\alpha+1)}
 C_{2n}^{(1+\alpha)}\left(\frac{a^2+b^2}{a^2-b^2}\right),
\end{eqnarray}
which agrees with \eqref{OP} for even indices. The proof for the odd polynomials follows exactly in the same way, and for completeness we have collected the necessary steps in Appendix \ref{La-odd}.
 \end{proof}

\begin{rem}\label{remarkU}
 In the case  $\alpha=0$  we recover  the orthogonality relation for Chebyshev polynomials of the second kind, due to $U_n(x)=C_n^{(1)}(x)$, which goes back to  \cite{Henrici}.  
We will come back to this statement in Section \ref{sec:Jacobi}.

For  $\alpha= -1/2$ we obtain as a special case the orthogonality of the 
Legendre polynomials $P_n(x)=C_n^{(1/2)}(x)$:
\begin{cor}\label{cor:Legendre}
The Legendre polynomials $P_n$ are orthogonal with respect to the weight function $\d A_\alpha$ defined in \eqref{ellipseweight} at $\alpha=-1/2$: 
\begin{equation}
\label{Legendre}
\Eint P_{m}\left(\frac{z}{c}\right)P_{n}\left(\frac{\overline{z}}{c}\right)\,\d A_{-\frac12}{(z)}= \frac{1}{1+2n}P_n\le\frac{a^2+b^2}{a^2-b^2}\ri\delta_{n,m}\ .
\end{equation}
\end{cor}
We have not been able to find this result in the literature.

Furthermore, we can make contact with  Hermite polynomials\footnote{We denote by this the 
Hermite polynomials orthogonal with respect to $\exp[-x^2]$ on $\mathbb{R}$.}  $H_n(x)$ as polynomials in the full complex plane.  Setting $z\rightarrow z/\sqrt{1+\alpha}$ and taking $\alpha$ to infinity in \eqref{OP}, we have from \cite[18.7.24]{NIST}
\begin{equation}
\lim_{\alpha\to\infty}
(1+\alpha)^{-\frac{n}{2}}
C_n^{(1+\alpha)}\left(
(1+\alpha)^{-\frac{1}{2}}
x\right)=H_n(x)/n!\ ,
\end{equation}
 leading to 
\begin{eqnarray} \label{finalhermite}
\int_{\mathds{C}}\,H_m(z/c)\,H_n(\overline{z}/c)\,e^{-h(z)}\d^2{z}&=& \pi n!\ ab\le 2\frac{a^2+b^2}{a^2-b^2}\ri^n\delta_{n,m}\ ,
\end{eqnarray}
with $h(z)$ defined in \eqref{ellipsedef}.
This reproduces the known orthogonality relation for Hermite polynomials in the complex plane, obtained by van Eijndhoven and Meyers \cite[Eq.(0.5)]{EM} 
for  $a=\sqrt{\frac{1}{1-A}}$ and $b=\sqrt{\frac{A}{1-A}}$, with $0<A<1$, see also \cite{PdF}.
\end{rem}
 
 \begin{rem}
 \label{rem:tUnitary}
In the limit $c\to0$, when the ellipse $E$ becomes a disc, we obtain for integer values of $\alpha$ the weight function that results from the complex eigenvalues of the ensemble of truncated unitary random  matrices studied in \cite{ZS}, with monomials as orthogonal polynomials. This can be seen as follows:
 We have from eq. \eqref{eq:par} that the monic Gegenbauer polynomials occurring in \eqref{OP} read:
 \begin{equation}\label{eq:Cmonic}
\tilde{p}^{(\alpha)}_n(z):=\frac{n!c^n}{2^n (1+\alpha)_{n}}C_{n}^{(1+\alpha)}(z/c)\ .
 \end{equation}
Multiplying \eqref{OP} with the corresponding factors we can take the limit $b\to a$,  implying $c\to 0$ 
in  this orthogonality relation, to obtain
\begin{equation}
\label{truncatedU}
\int_{x^2+y^2<a^2} z^m\bar{z}^n(1+\alpha)\left(1-\frac{|z|^2}{a^2} \right)^\alpha\frac{\d^2z}{\pi a^2}\ =\ \frac{\Gamma(n+1)\Gamma(1+\alpha)(1+\alpha)}{\Gamma(1+\alpha+n)(1+\alpha+n)}a^{2n}\delta_{n,m}\ ,
\end{equation}
where $z=x+iy$. After rescaling $z\to az$, and dividing \eqref{truncatedU} by $(1+\alpha)$, we arrive at the weight function and monic polynomials for the complex eigenvalues in the ensemble of truncated unitary random matrices \cite{ZS} on the unit disc. It is defined starting from the circular unitary ensemble of Haar distributed unitary random matrices of size $N\times N$ and truncating these to the upper left block of size $M\times M$ with $N>M$, by removing $N-M$ rows and columns. The weight function reads $w(z)=(1-|z|^2)^{N-M-1}$, that is we have to identify  $\alpha=N-M-1\geq0$. In this case there is no singularity on the boundary of the circle, and we may extend our integration from inside  the disc to include the boundary, cf. \eqref{intbarE}.
In this ensemble this is important as for large $M$ and small truncation $N-M$ a substantial fraction of eigenvalues of the truncated unitary matrix may remain on the unit circle. We refer to \cite{ZS} for a further discussion of the limiting behaviour.

In analogy to the relation between the Ginibre ensemble and its elliptic version, our Gegenbauer polynomials can thus be viewed as the orthogonal polynomials of an elliptic version of the truncated unitary ensemble \cite{ZS}, with an appropriate random matrix realisation yet to be constructed.  
 \end{rem}

\begin{rem}\label{reallim} 
Finally, we can establish contact with the usual orthogonality relation for the Gegenbauer polynomials on the real interval $[-1,1]$.
The change of variables for the imaginary part $y =\frac{b}{a}\hat{y}$ maps the ellipse to a disc of radius $a$. 
Together with $\d A_\alpha(z)=(1+\alpha)(1-(x/a)^2-(y/b)^2)^\alpha\d x\d y/(ab\pi)$,  this allows us to take the limit $b \rightarrow 0$ on \eqref{OP}
\begin{eqnarray}
&&\lim_{b\to0}\Eint C_{m}^{(1+\alpha)}\left(\frac{z}{c}\right)C_{n}^{(1+\alpha)}\left(\frac{\overline{z}}{c}\right)\,\d A_\alpha{(z)}= \no\\
&&=\int_{-a}^aC_{m}^{(1+\alpha)}\left(\frac{x}{a}\right)C_{n}^{(1+\alpha)}\left(\frac{x}{a}\right)\left(1-\frac{x^2}{a^2}\right)^\alpha
\int_{-\sqrt{a^2-x^2}}^{\sqrt{a^2-x^2}}\left(1-\frac{\hat{y}^2}{a^2-x^2}\right)^\alpha
\frac{(1+\alpha)\d \hat{y}\d x  }{a^2\pi} \no\\
&&=\int_{-1}^1C_{m}^{(1+\alpha)}(x)C_{n}^{(1+\alpha)}(x)\left(1-x^2\right)^{\alpha+\frac12}F\left(\frac12,-\alpha;\frac32;1\right)
\frac{2(1+\alpha)}{\pi}\d x
\no\\
&&=\frac{1+\alpha}{1+\alpha+n}C_{n}^{(1+\alpha)}(1)\delta_{n,m}\ .\no\\
\label{Climbto0}
\end{eqnarray}
Identities for the Gegenbauer polynomial \cite[Table 18.6.1]{NIST} 
\begin{equation}
C_{n}^{(1+\alpha)}(1)=\frac{\Gamma(2+2\alpha+n)}{\Gamma(2+2\alpha)\Gamma(n+1)}\ ,
\end{equation}
and for Gau{\ss}' hypergeometric function \cite[9.122]{Grad} at unity  
\begin{equation}
F\left(\frac12,-\alpha;\frac32;1\right)=\frac{\sqrt{\pi}\Gamma(1+\alpha)}{2\Gamma(\alpha+3/2)}\ ,
\end{equation}
yield the standard orthogonality relation 
\begin{eqnarray} 
\;\;\;\int_{-1}^{1}\,C_{n}^{(\alpha+1)}(x)\, C_{m}^{(\alpha+1)}(x)(1-x^2)^{\alpha +\frac{1}{2}}\d{x}
&=& \frac{2^{1-2(1+\alpha)}\pi\G(2+2\alpha+n)}{(1+\alpha+n)\G^2(1+\alpha)n!}\delta_{n,m}\ .
\label{real}
\end{eqnarray}
\end{rem}

We can now finish the  proof of Theorem \ref{propONB} by showing the completeness of the system of orthogonal polynomials.
 \begin{proof}
 Let $f \in A_{\alpha}^2$ with $\langle f,p_n\rangle_{\alpha} =0$ for all $n=0,1,2,...$. Then 
 \be\label{completeness}
 0= \lim_{b \to 0}\langle f,p_n\rangle_{\alpha}=  \int_{-1}^{1}\d{x}\,f(ax)\,C_{n}^{(1+\alpha)}(x)(1-x^2)^{\alpha +\frac{1}{2}}\ .
 \ee 
 Hence $f(ax)=0$ for all $x \in (-1,1)$, see \cite{Szego} for the completeness of the Jacobi polynomials on the real line. Since $f$ is regular in $E$, it follows that $f\equiv0$, i.e. $\{p_n^\alpha\}$ defined above form an orthonormal basis for $A_{\alpha}^{2}$.
 \end{proof}


\section{Orthogonality of certain Jacobi and all Chebyshev polynomials}
\label{sec:Jacobi}

In this section we will first deduce Corollary \ref{CorrJabobi} from our Lemma \ref{COP-lemma}, by mapping the Gegenbauer polynomials to a certain sub-family of Jacobi polynomials $P_n^{(\alpha+\frac12,\pm\frac12)}$ orthogonal on an ellipse, with a different weight function. 
We will not use the standard, symmetric representation \cite[18.7.1]{NIST}
\begin{equation}
\label{CPlinear}
C_n^{(1+\alpha)}(z)=\frac{(2+2\alpha)_n}{(\alpha+\frac32)_n}P_n^{(\alpha+\frac12,\alpha+\frac12)}(z)\ ,
\end{equation}
which is linear, but rather a quadratic transformation that leads to a non trivial orthogonality relation, as described below. 
Second, we will use this corollary to show the orthogonality of Chebyshev polynomials of the first, second, third and fourth kind 
$T_n$, $U_n$, $V_n$ and $W_n$, respectively, 
that were derived in a different way in \cite{Mason}, see
Corollary \ref{CorrCheby} below.
We will come back to the polynomials $U_n$ of the second kind, where the orthogonality was already stated in Remark \ref{remarkU}, following from Lemma \ref{COP-lemma}.

Let us summarise our first statement as follows.
\begin{cor}\label{CorrJabobi}
Define the ellipse $E$ as before in \eqref{ellipsedef},
and the function
\begin{equation}
\label{jdef}
j(w)= \frac{a}{b^2}|c+w|-\frac{c}{b^2}\Re{(c+w)}\ .
\end{equation}
It satisfies $0<j(w)<1$ on $E$. Then, for $\alpha>-1$ the following two sub-families of Jacobi polynomials are orthogonal on $E$: First,
 \begin{eqnarray} \label{OP2n}
\int_{E} P_n^{(\alpha+\frac12,-\frac12)}(w/c)P_m^{(\alpha+\frac12,-\frac12)}(\overline{w}/c)\,\d B_\alpha^-{(w)}&=& \frac{((1/2)_n)^2}{((\alpha+1)_n)^2}
\frac{1+\alpha}{1+\alpha+2n}C_{2n}^{(1+\alpha)}\le\frac{a}{c}\ri\delta_{n,m},\no\\
\end{eqnarray}
with respect to the weight function
\begin{equation}
  \label{dB-}
 \d{B}_{\alpha}^-(w)\,:=\,\frac{(1+\alpha)}{2\pi b}\frac{(1-j(w))^{\alpha}}{|c+w|}
 \d^2 w\ ,
\end{equation}
and, second, 
\begin{eqnarray} \label{OP2n+1}
\int_{E} P_n^{(\alpha+\frac12,\frac12)}(w/c)P_m^{(\alpha+\frac12,\frac12)}(\overline{w}/c)\,\d B_\alpha^+{(w)}&=& 
\frac{2c(1/2)_{n+1}^2(1+\alpha)(2+\alpha)}{a(\alpha+1)_{n+1}^2(2+\alpha+2n)}C_{2n+1}^{(1+\alpha)}\le\frac{a}{c}\ri\delta_{n,m},\no\\
\end{eqnarray}
with weight 
\begin{equation}
\label{dB+}
 \d{B}_{\alpha}^+(w)\,:=\,\frac{(1+\alpha)(2+\alpha)}{2\pi ab}(1-j(w))^{\alpha}\d^2 w\ .
\end{equation}
The measures $\d{B}_\alpha^\pm(w)$ are chosen to be normalised, $\int_{E}\d{B}_\alpha^\pm(w)=1$. 
 \end{cor}
\begin{proof}
We begin with the orthogonality relation \eqref{OP2n}.  Using the quadratic transformation \cite[18.7.15]{NIST}, we have for all even Gegenbauer polynomials 
\begin{equation}
\label{C2nP}
C_{2n}^{(\alpha+1)}(z/c)=\frac{(\alpha+1)_n}{\left(\frac12\right)_n} P_n^{(\alpha+\frac12,-\frac12)}\left(2\left( \frac{z}{c}\right)^2-1\right)\ ,
\end{equation}
which leads us to identify the map to a new coordinate $w$
\begin{equation}
\label{wzmap}
\frac{w}{c}=2\left( \frac{z}{c}\right)^2-1\ .
\end{equation}
In order to specify the orthogonality relation following from Lemma \ref{COP-lemma} for this family of Jacobi polynomials in $w/c$, we have to determine the domain and weight function resulting from the map \eqref{wzmap} of the ellipse \eqref{ellipsedef} and weight \eqref{ellipseweight}. In order to make the mapping \eqref{wzmap} to be inverted for $z$ unique,
\begin{equation}
\label{zwmap}
\frac{z(w)}{c}=\sqrt{\frac{w+c}{2c}}\ ,
\end{equation}
we subdivide the ellipse $E=E_+\cup E_-$, with 
\begin{equation}\label{Eplus}
E^{+(-)}\,=\, \{z=x+iy \in \mathds{C}: x^2/a^2+y^2/b^2<1, x>(<)0 \} \ .
\end{equation}
Because the weight and measure are invariant under the inversion $z\to-z$, that maps $E^-\to E^+$, we have for the even Gegenbauer polynomials before the map \eqref{zwmap}
\begin{eqnarray}
\label{Epmmap}\quad\quad
\int_E C_{2n}^{(\alpha+1)}(z/c)C_{2l}^{(\alpha+1)}(\bar{z}/c)\d A_\alpha{(z)}&=& 2\int_{E^+} C_{2n}^{(\alpha+1)}(z/c)C_{2l}^{(\alpha+1)}(\bar{z}/c)\d A_\alpha{(z)}
\\
&=&\int_{\tE} P_n^{(\alpha+\frac12,-\frac12)}(w/c)P_m^{(\alpha+\frac12,-\frac12)}(\overline{w}/c)\,\d \tilde{B}_\alpha^-{(w)},
\nonumber
\end{eqnarray}
which leads to the claimed orthogonality in the second line, as we will explain now.
The map \eqref{zwmap} has a square root cut for $\{\Re(w)+c<0\}$, and for the Jacobian we obtain from 
\begin{equation}
\frac{\d z(w)}{\d w}=\frac{\sqrt{c}}{2\sqrt{2}}\frac{1}{\sqrt{w+c}}\ \ \Rightarrow\ \ \d^2 z = \frac{c\,\d^2 w}{8|w+c|}\ .
\label{zwJac} 
\end{equation}
In order to determine the domain to be $\tE\subset \mathbb{C}\setminus\{\Re(w)+c<0\}$, resulting from \eqref{zwmap}, we introduce two auxiliary quantities $A>B>0$ in terms of the parameters $a>b>0$ of the original ellipse $E$ in \eqref{ellipsedef}:
\begin{equation}
\label{ABdef}
A=\frac{a^2+b^2}{2a^2b^2}\ ,\ \ B=\frac{a^2-b^2}{2a^2b^2}=\frac{c^2}{2a^2b^2}\ .
\end{equation} 
Here, we have recalled the definition of parameter $c$. These two quantities satisfy
\begin{equation}
\label{ABrel}
A^2-B^2=\frac{2B}{c^2}\ \ \mbox{and}\ \ \frac{Bc^2}{2}+1=\frac{(a^2+b^2)^2}{4a^2b^2}\ .
\end{equation}
Furthermore, we can write for $z=x+iy$
\begin{equation}
\label{oldell}
A|z|^2-B\Re(z^2)=\frac{(a^2+b^2)(x^2+y^2)}{2a^2b^2}-\frac{(a^2-b^2)(x^2-y^2)}{2a^2b^2}=\frac{x^2}{a^2}+\frac{y^2}{b^2}\ .
\end{equation}
Therefore, the domain $E$ expressed in terms of the new variable $w=u+iv$ reads
\begin{eqnarray}
\label{newell}
1> A|z|^2-B\Re(z^2)&=&\frac{cA}{2}|w+c|-\frac{cB}{2}\Re(w+c)\no\\
&=&\frac{cA}{2}\sqrt{(u+c)^2+v^2}-\frac{cB}{2}(u+c),
\end{eqnarray}
which is the defining equation for the new domain. The claim that it is again given by an ellipse, with new parameters $\ta$ and $\tb$ to be determined, can be seen as follows. From \eqref{newell} we have
\begin{eqnarray}
&&0<\frac{cA}{2}\sqrt{(u+c)^2+v^2}< 1+\frac{cB}{2}(u+c) \no\\
\Leftrightarrow&&\frac{c^2}{4}(A^2-B^2)u^2+uc\left( \frac{c^2(A^2-B^2)}{2}-B\right)+ \frac{c^2A^2}{4}v^2< 1-\frac{c^4(A^2-B^2)}{4}+Bc^2\no\\
\Leftrightarrow&& c^2u^2+\frac{c^2(a^2+b^2)^2}{4a^2b^2}v^2< (a^2+b^2)^2\ ,\no\\
\end{eqnarray}
which is obtained after squaring the inequality, using \eqref{ABrel} and multiplying with $4a^2b^2$. We are thus led to define the new domain as 
\begin{equation}
\label{new ellipse}
\tE:=\{w=u+iv\in\mathbb{C}: u^2/\ta^2+v^2/\tb^2<1\}\ ,\ \mbox{with}\ \ \ta=\frac{(a^2+b^2)}{c}\ ,\ \ \tb=\frac{2ab}{c} \ .
\end{equation}
We note that $c^2=\ta^2-\tb^2=a^2-b^2$ follows. It remains to show \eqref{jdef}, which follows from \eqref{oldell} and \eqref{newell} as 
\begin{equation}\label{hjeq}
1-h(z)=1-\frac{cA}{2}|w+c|+\frac{cB}{2}\Re(w+c)=1-\tilde{j}(w)\ ,
\end{equation}
together with $cA/2=\ta/\tb^2$ and $cB/2=c/\tb^2$. Inserting all these into the right hand side of \eqref{Epmmap}, we arrive at \eqref{OP2n} with weight \eqref{dB-}.
The fact that the weight  \eqref{dB-} is normalised to unity immediately follows from setting 
$n=0$ in \eqref{OP2n}, as $C_0^{(1+\alpha)}(x)=1$. Dropping the tilde on all quantities  we arrive at the statement in \eqref{OP2n}.

It is important to note here that once we consider the Jacobi polynomials \eqref{OP2n}, with weight \eqref{dB-} on the new ellipse \eqref{new ellipse}, there is no more square root cut inside, which would lead to a slit domain. To see this we multiply the equation defining $E$
\begin{eqnarray}\label{jeq}
&&\frac{u^2}{a^2}+\frac{v^2}{b^2}<1\no\\
\Leftrightarrow&& a^2\left((u+c)^2+v^2\right)<(b^2+c(u+c))^2\ ,
\end{eqnarray}
by $a^2b^2$, to arrive at the second line. While it is clear that the left hand side is always positive, we can take the square root here without crossing zero, due to the following fact. It holds that $b^2+c(u+c)=cu+a^2$ inside the square on the right hand side is always positive for $u\in(-a,+a)$. This ends the proof for the first set of polynomials.

For the orthogonality relation \eqref{OP2n+1} of the second set of polynomials only few modifications are needed. We start from  Lemma \ref{COP-lemma} for the odd Gegenbauer polynomials and use the relation \cite[18.7.16]{NIST}:
\begin{eqnarray}
\label{C2n+1P}
C_{2n+1}^{(\alpha+1)}(z/c)&=&\frac{(\alpha+1)_{n+1}}{\left(\frac12\right)_{n+1}} \frac{z}{c} P_n^{(\alpha+\frac12,\frac12)}(2(z/c)^2-1)\ .
\end{eqnarray}
The identification of variables \eqref{wzmap} is identical, and the map to $E^+$ works in the same way as in \eqref{Epmmap}, after cancelling the two minus signs obtained from the reflection of the two odd polynomials. 
Apart from the additional constant factors, we obtain from \eqref{C2n+1P} an additional factor
\begin{eqnarray}
\left| \frac{z}{c}\right|^2=\frac{|w+c|}{2c}
\end{eqnarray}
which cancels the pole from the Jacobian in \eqref{zwJac}. This leads to the orthogonality \eqref{OP2n+1} with weight \eqref{dB+}, after multiplying with an overall factor $(2+\alpha)$ for the correct normalisation of the area measure. This can be seen using that 
$C_1^{(1+\alpha)}(x)=2(1+\alpha)x$ for $n=0$ on the right hand side.
\end{proof}
\begin{rem}
In order to show like in Section \ref{Complete metric spaces} that the Bergman space with weights $\d B_\alpha^-$  \eqref{dB-} and $\d B_\alpha^+$ \eqref{dB+} is closed in $L^p$, all one needs to do is to find an estimate as in \eqref{sup}, such that we can apply Proposition \ref{bound} and Corollary \ref{cor:closed} to these. 

From \eqref{jeq} it follow that $j(z)=1$ if and only if $z \in \partial{E}$. It is easy to see that there are no local extrema  for $j(z)$ inside of ${E}_{\rho}$, therefore $0< \max_{z\in {E}_{\rho}} j(z)=j(z_{*})<1$ for some $z_{*} \in {E}_{\rho}$. This shows that the Bergman space  $A_{\alpha}^p(E,\d B_\alpha^+)$ is closed in $L^p(E,\d B_\alpha^+)$. As a consequence of H\"{o}lder's inequality we obtain the same result for the weight  $\d B_\alpha^-$ and $p\geq 1$.
\end{rem}

\begin{rem}
As it was done in Remark \ref{reallim} we can make contact to the usual orthogonality relations for Jacobi polynomials $P_{n}^{(\alpha+\frac{1}{2},\frac{1}{2})}$ on the real line, by rescaling the imaginary part of $z$,  $\im z\to \frac{b}{a}\im z$.
The same steps can be taken for $P_{n}^{(\alpha+\frac{1}{2},-\frac{1}{2})}$. Without giving any details, in the limit $b\to 0$ we obtain in analogy to \eqref{Climbto0}
\begin{eqnarray}
&&\lim_{b\to0}\Eint P_{m}^{(\frac{1}{2}+\alpha,\frac{1}{2})}\left(\frac{z}{c}\right)P_{n}^{(\frac{1}{2}+\alpha,\frac{1}{2})}\left(\frac{\overline{z}}{c}\right)\,\d B_\alpha^+{(z)}= \no\\
&&=F\left(\frac12,-\alpha;\frac32;1\right)
\frac{(1+\alpha)(2+\alpha)}{2^{\alpha}\pi}\int_{-1}^1P_{m}^{(\frac{1}{2}+\alpha,\frac{1}{2})}(x)P_{n}^{(\frac{1}{2}+\alpha,\frac{1}{2})}(x)\left(1-x\right)^{\alpha+\frac12}(1+x)^{\frac12}\d x
\no\\
&&=\frac{2(1/2)_{n+1}^2(1+\alpha)(2+\alpha)}{(\alpha+1)_{n+1}^2(2+\alpha+2n)}C_{2n+1}^{(1+\alpha)}\le1\ri\delta_{n,m}\ ,\no\\
\end{eqnarray}
which yields the correct normalisation on $[-1,1]$, see \cite[18.3.1]{NIST}.

Finally, like in \eqref{completeness}, it is easy to see that  Jacobi polynomials $P_n^{(\alpha+\frac12,\pm\frac12)}$ provide an orthonormal basis for the underlying  Hilbert space from Corollary \ref{CorrJabobi}.
\end{rem}

In the remaining part of this section we will prove the orthogonality of the Chebyshev polynomials of first to fourth kind as a direct consequence of Corollary 
\ref{CorrJabobi}. The following statement is due to \cite{Mason}, 
 where the notation for the polynomials of third and fourth kind is interchanged compared to ours, $V_n\leftrightarrow W_n$. We follow the notation of \cite{NIST}.

\begin{cor}\label{CorrCheby}
The Chebyshev polynomials satisfy the following orthogonality relations  on the ellipse defined in \eqref{ellipsedef}, with $r=a+b$ and $c^2=a^2-b^2$:
\begin{eqnarray}
\label{Ch1st}
\quad\quad\quad \Eint T_{n}(z/c)T_{m}(\overline{z}/c)\,\frac{\d^2z}{|z^2-c^2|}&=& 
\left\{
\begin{array}{ll}
\displaystyle{\frac{\pi}{4n}}
((r/c)^{2n}-(c/r)^{2n})\delta_{n,m}& \mbox{for}\;  n>0,m\geq0,\\
&\\
\displaystyle{2\pi \ln({r}/{c})}
& \mbox{for}\;  n=m=0\ ,\ \\
\end{array}
\right.\\
\label{Ch2nd}
\Eint U_n(z/c)U_m(\overline{z}/c)\,\d^2z
&=&\frac{\pi c^2}{4(1+n)}((r/c)^{2n+2}-(c/r)^{2n+2})\delta_{n,m}\ ,\\
\label{Ch3rd}
\Eint V_n(z/c)V_m(\overline{z}/c)\,\frac{\d^2 z}{|c+z|}&=& \frac{\pi c}{1+2n}
((r/c)^{2n+1}-(c/r)^{2n+1})\delta_{n,m}\ ,\\
\label{Ch4th}
\quad\Eint W_n(z/c)W_m(\overline{z}/c)\,\frac{\d^2 z}{|c-z|}&=& \frac{\pi c}{1+2n}((r/c)^{2n+1}-(c/r)^{2n+1})\delta_{n,m}\ .
\end{eqnarray}
\end{cor}

Note that for better comparison with \cite{Mason}\footnote{In contrast to the orthogonality of the Chebyshev polynomials on the contour given by the boundary of the ellipse $\partial E$ stated in \cite{Mason} too, the weight function we find here differs from the classical weight on the real line, continued to the ellipse.} our statements are with respect to the flat measure $\d^2z$, rather than the area measure 
$\d A(z)=\d^2z/(\pi ab)$.

\begin{proof}
We begin with the Chebyshev polynomials of the second kind $U_n$.
Because of the relation \cite[18.7.4]{NIST}
\begin{eqnarray}\label{UCP}
U_n(z)=C^{(1)}_n(z)= \frac{1+n}{P_n^{(1/2,1/2)}(1)} P_n^{(1/2,1/2)}(z)\ ,
\end{eqnarray}
we set $\alpha = 0$ in \eqref{OP2n+1}, to obtain
 \begin{eqnarray} \label{Un}
\Eint P_n^{(1/2,1/2)}(z/c)P_m^{(1/2,1/2)}(\overline{z}/c)\frac{\d^2z}{\pi ab}&=&
\frac{2c((1/2)_{n+1})^2}{a((1)_{n+1})^2(1+n)}C_{2n+1}^{(1)}\le\frac{a}{c}\ri\delta_{n,m},\no\\
\end{eqnarray}
After using 
\begin{equation}
\label{C1norm}
P_n^{(\frac12,\frac12)}(1)=\Gamma(n+3/2)/(\Gamma(3/2)\Ga(n+1))\ ,
\end{equation}
from \cite[Table18.6.1]{NIST}, we arrive at 
\begin{eqnarray} \label{OP2nparticular2}
\Eint U_n(z/c)U_m(\overline{z}/c)\,\d^2z&=& \frac{\pi cb}{2(1+n)}C_{2n+1}^{(1)}\le\frac{a}{c}\ri\delta_{n,m}\ .
\end{eqnarray}
Recalling $r=a+b$ and $c^2=a^2-b^2$, we have
\begin{equation}
\frac12\left( \frac{r}{c}+\frac{c}{r}\right)=\frac{a}{c}\ \ \mbox{and}\ \ \frac12\left( \frac{r}{c}-\frac{c}{r}\right)=\frac{b}{c}\ .
\end{equation}
With this the Gegenbauer polynomial on right hand side of \eqref{OP2nparticular2} can be simplified as follows. Applying \eqref{CcosId} we have 
\begin{eqnarray}
C_{2n+1}^{(1)}(\cos(i\ln(r/c)))
&=&\sum_{k=0}^{2n+1}\cos((2n+1-2k)i\ln(r/c))\no\\
&=& \sum_{k=0}^{2n+1}\frac12 \left( (c/r)^{2n+1-2k}+(r/c)^{2n+1-2k}\right) 
\no\\ 
&=&\frac{c}{2b} ((r/c)^{2n+2}-(c/r)^{2n+2})\ .
\label{C2n+1simple}
\end{eqnarray}
When replacing $C_{2n+1}^{(1)}({a}/{c})$ in \eqref{OP2nparticular2} we arrive at the statement \eqref{Ch2nd}.

Let us recall that the orthogonality of $U_n(z)$ \eqref{UCP} also follows by setting $\alpha=0$ in Lemma \ref{COP-lemma}, see Remark \ref{remarkU}. Comparing this statement 
\begin{equation}
\Eint U_n(z/c)U_m(\overline{z}/c)\,\d^2z= \frac{\pi ab}{(1+n)}C_{n}^{(1)}\le\frac{a^2+b^2}{a^2-b^2}\ri\delta_{n,m}\ ,
\end{equation}
with \eqref{OP2nparticular2}, we find that the following identity must hold:
\begin{equation}
\label{CC}
C_{2n+1}^{(1)}(x) = x C_n^{(1)}(2x^2-1)\ .
\end{equation}
Indeed this follows from the quadratic relation \eqref{C2n+1P} at $\alpha=0$, and \eqref{UCP}.
We emphasise, however, that beyond $\alpha=0$ apparently no such identity \eqref{CC} exists, that would allow to further simplify the right-hand side of Lemma \ref{COP-lemma}.

The Chebyshev Polynomials of the third kind $V_n$ are related to Jacobi polynomials following \cite[18.7.5]{NIST}:
 \begin{equation}
 V_n(z/c)= \frac{1+2n}{P_n^{(1/2,-1/2)}(1)} P_n^{(1/2,-1/2)}(z/c)\ .
 \end{equation}
Setting $\alpha=0$ in (\ref{OP2n}) we obtain
 \begin{equation} \label{OP2nparticular1}
\Eint V_n(w/c)V_m(\overline{w}/c)\,\frac{\d^2 w}{|c+w|}= \frac{2\pi b}{1+2n}C_{2n}^{(1)}\le\frac{a}{c}\ri\delta_{n,m}\ .
\end{equation}
Here, we have inserted \eqref{C1norm}.
Similar to \eqref{C2n+1simple} we can simplify the Gegenbauer polynomial on right hand side of \eqref{OP2nparticular1}, using  \eqref{CcosId} for an even index. We have  
\begin{eqnarray}
C_{2n}^{(1)}(\cos(i\ln(r/c)))
&=&\sum_{k=0}^{2n}\cos((2n-2k)i\ln(r/c))\no\\
&=& \sum_{k=0}^{2n}\le \frac{r}{c}\ri^{2n-k}\le \frac{c}{r}\ri^k\no\\ 
&=&\frac{c}{2b} ((r/c)^{2n+1}-(c/r)^{2n+1})\ ,
\label{C2nsimple}
\end{eqnarray}
which upon
replacing $C_{2n}^{(1)}({a}/{c})$ in \eqref{OP2nparticular1} leads to the statement \eqref{Ch3rd}.

The orthogonality relation for  Chebyshev polynomials of the fourth kind $W_n$ is simple, due to the relation $V_n(-x)=(-1)^nW_n(x)$ true $\forall n\in\mathbb{N}$. A reflection 
$z\to-z$ upon \eqref{OP2nparticular1} leads to 
\begin{equation}
\label{Ws}
\Eint W_n(w/c)W_m(\overline{w}/c)\,\frac{\d^2 w}{|c-w|}
= \frac{2\pi b}{1+2n}C_{2n}^{(1)}\le\frac{a}{c}\ri\delta_{n,m}\ .
\end{equation}
Here, the signs trivially cancel due to $\delta_{n,m}$, together with the simplification \eqref{C2nsimple} just described, leading to \eqref{Ch4th}.

We turn to the orthogonality for the Chebyshev polynomials of the first kind $T_n$. The relation \cite[18.7.18]{NIST}
\begin{equation}
T_{2n+1}(x)=xW_n(2x^2-1)
\end{equation}
allows us to find the corresponding weight function and orthogonality of the odd polynomials, starting from \eqref{Ws}:
\begin{eqnarray} \label{Cheb}
\int_{\tilde{E}}W_n(z'/c)W_m(\overline{z'}/c)\,\frac{\d^2z'}{|z'-c|}
&=&8c \int_{E^{+}}\frac{z}{c}W_n(2(z/c)^2-1)\frac{\overline{z}}{c}W_m(2(\overline{z}/c)^2-1)\frac{\d^2z}{|z^2-c^2|}\no\\
&=&4c \Eint T_{2n+1}(z/c)T_{2m+1}(\overline{z}/c)\,\frac{\d^2z}{|z^2-c^2|}\ .
\end{eqnarray}
Here, we use the inverse transformation of \eqref{wzmap} applied in the proof of Corollary \ref{CorrJabobi}, see $E^{+} $ and $ \tilde{E}$ are defined there. 
Thus
the  polynomials $\{T_n\}$ are orthogonal w.r.t. $\frac{1}{|z^2-c^2|}\d^2 z$. 
The following well known relation \cite{Mason} holds for the 
Joukowsky map $z/c=\frac{1}{2}(w/c+c/w)$
\begin{equation}
\label{TJ}
T_n(z/c)=\frac{1}{2}((w/c)^n+(c/w)^{n})\quad \text{for}\quad n\geq 0,
\end{equation}
which maps the ellipse $E$ to the annulus  $A:=\{w \in \mathds{C}: c<|w|<r\}$. 
We thus obtain for $n>0,m\geq0$
\begin{eqnarray} 
&&\Eint T_{n}(z/c)T_{m}(\overline{z}/c)\,\frac{\d^2z}{|z^2-c^2|}\no\\
&=& \int_{A}T_{n}(z(w)/c)T_{m}(\overline{z(w)}/c)\frac{\d^2w}{|w|^2} \no\\
&=& \frac14 \int_{c}^r \frac{\d s}{s} \int_0^{2\pi}\d \theta
\left( (s/c)^ne^{in\theta}-(c/s)^ne^{-in\theta}\right) 
\left( (s/c)^me^{-im\theta}-(c/s)^me^{im\theta}\right) \no\\
&=& \frac{\pi}{4n}((r/c)^{2n}-(c/r)^{2n})\delta_{n,m}\no\\
&=&\frac{\pi b}{2nc}C_{2n-1}^{(1)}\left(\frac{a}{c}\right)\delta_{n,m}\ ,
\label{Cheb1}
\end{eqnarray} 
by changing to polar coordinates $w=s\,e^{i\theta}$. Performing the elementary integrations we need to restrict us to $n>0,m\geq0$. The first part of \eqref{Ch1st} follows and  in the last step we have inserted \eqref{C2n+1simple}, in order to compare to the previous orthogonality relations. 
For $n=m=0$ with $T_0(x)=1$, following the same computation we have
\begin{eqnarray} \label{Cheb10}
\Eint \frac{1}{|z^2-c^2|}\d^2z&=& 2\pi \ln(r/c)\ ,
\end{eqnarray} 
which ends the proof of Corollary \ref{CorrCheby}.
\end{proof}


\section{Bergman polynomials, Selberg integrals and finite-term recurrence
}\label{Bergman-recursion}

All the orthogonal polynomials on an ellipse we encountered in the previous section satisfy a three-step recursion relation, as they result from classical polynomials on the real line. 
For polynomials on a planar, sufficiently regular domain with flat weight function this is a generic feature as summarised in the Theorem \ref{thm:KS} of Khavinson and Stylianopoulos  \cite{khavinson} below.
Using the  Gegenbauer polynomials from Lemma \ref{COP-lemma} that are orthogonal on a weighted ellipse, we will construct an example, that on a weighted domain this statement is no longer true, invalidating  the finite-term recursion relation. Because we will work with normalised expectation values to construct such an example, we will state in passing the normalising factor (partition function) for Gegenbauer polynomials, constituting a special case of a Selberg integral in the complex plane.

Consider a  bounded simply connected domain $D$ in the complex plane, let $\d{\mu}(z)=w(z)\d{A}(z)$  be a measure on $D$, where $\d A$ is the planar Lebesgue measure, and 
$w$  a non-negative weight function on $D$.
Given that all moments exist, $\int_D z^k \bar{z}^l w(z)\,\d{A}(z)<\infty$, 
 a unique sequence of polynomials $p_n(z)=\gamma_n z^n+... ; \gamma_n>0$ can be constructed using the Gram-Schmidt process,  that are orthonormal w.r.t $\d{\mu}$, see e.g. \cite{Walter}: 
\begin{equation}
\label{pOP}
\int_D p_n(z)\overline{p_m(z)} w(z)\,\d{A}(z)=\delta_{n,m}\ .
\end{equation} 
In the literature these polynomials are called {\it Bergman orthonormal polynomials}. For example  choosing $D=E$ as an ellipse, and $w(z)=(1+\alpha)(1-h(z))^\alpha$, these polynomials  $p_n$ are proportional to the Gegenbauer polynomials, see \eqref{bin}.
 
The multiplication operator acting on polynomials can always be represented by expanding $z\,p_n(z)$ as a series of the Bergman polynomials being a basis: 
\be
\label{mult}
z \,p_n(z)  \ = \ \sum_{l=0}^{n+1} c_{l,n} \,p_l(z)\;, n=0,1,2,\ldots
\ee
The Fourier coefficients $c_{l,n} $ are then given by 
\be\label{Fourrier}
c_{l,n}\ =\ 
\Eint z\,p_n(z) \overline{ p_l(z) } w(z)\,\d A{(z)}\ .
\ee
These coefficients $c_{l,n} $ constitute the entries of an infinite upper Hessenberg matrix 
\[
 M=\begin{pmatrix}
  c_{0,0} & c_{0,1} & c_{0,2} & c_{0,3}& \dots\\
  c_{1,0} & c_{1,1} & c_{1,2} & c_{1,3}& \dots \\
   0         & c_{2,1}  & c_{2,2} & c_{2,3} & \dots \\
  0          & 0           &  c_{3,2} & c_{3,3} & \dots\\
  \vdots &\vdots    &\vdots & \ddots   & \ddots      
 \end{pmatrix}\ .
\]
This matrix provides a representation of  the Bergman Shift operator, 
which is defined by $(T_z\,f)(z)=zf(z)$ with respect to the basis $\{p_n\}_{n\in\mathbb{N}}$.

\begin{defn}[see \cite{khavinson}] We say that the upper Hessenberg matrix is {\it  banded} or, equivalently, that the orthogonal polynomials $p_n$ satisfy a { \it finite $(d+1)$-term recurrence} if there exists a positive integer $d$ such that 
\be
c_{l,n} = 0,\; \text{for}\; 0\leq l< n+1-d\ .
\ee
\end{defn}

In \cite{khavinson} Khavinson and Stylianopoulos proved the following 
\begin{thm}\label{thm:KS}
If the Bergman polynomials orthogonal with respect to the flat measure, on a {\it bounded simply-connected} domain $D$  with {\it regular enough} boundary, satisfy a $(d+1)$-{\it term recurrence} relation with $2\leq d$, then $D$ is an ellipse and $d=2$.
\end{thm}

For all orthogonal polynomials  supported on the real line, in particular for $I=[-1,1]$, it is well-known that the associated  orthogonal polynomials  satisfy a {\it three-term recurrence} relation $(d=2)$, including a nontrivial weight on $I$. Because the Gegenbauer polynomials that we found to be orthogonal on the weighted ellipse also satisfy a three-step recurrence, it is a natural question if the above Theorem \ref{thm:KS} extends to the weighted case. 
Unfortunately the answer is no, and we will construct a counter example below. 
The fact that in general in the complex plane no three-step recurrence can be expected was illuminated by \cite{Lempert}.

We will use an alternative representation to Gram-Schmidt that allows to construct orthogonal polynomials, the Heine formula. For a given domain $D\subseteq \mathds{C}$ in the complex plane,  a non-negative weight function  $w(z)$, and normalised area measure $\d{A}$ on $D$ such that all moment exist, we define the following expectation value:
\be
\langle\, {\mathcal O}\, \rangle_{N,w} \ =\mathcal{Z_N}^{-1} \int_{D^N}
 {\mathcal O}\left|\Delta_N(z)\right|^2
\prod_{i=1}^Nw(z_i) \d{A}(z_i)  \ ,
\label{eq:Zev}
\ee
where $\mathcal{O}$ depends on $z_{i=1,\ldots,N}\in\mathbb{C}$. 
Here, $\Delta_N(z)=\prod_{i>j}^N(z_i-z_j)$ is the Vandermonde determinant, and $\mathcal{Z}_N$ is a normalisation constant that ensures $\langle\, 1\, \rangle_{N,w}=1$.
The expectation value can be throught of resulting from the joint density of complex eigenvalues of a complex non-Hermitian random matrix ensemble, such as the elliptic Ginibre ensemble. The Heine formula then states that the orthogonal polynomials of degree $N$ in monic normalisation, $\tilde{p}_N(z)=z^N+\ldots$, are given by
\be
\tilde{p}_N(z)=\left\langle\, \prod_{i=1}^N(z-z_i)\, \right\rangle_{N,w} \ .
\label{Heine}
\ee
That is they are given by the expectation value of a single characteristic polynomial. Denoting the squared norms of the monic polynomials by $\tilde{h}_N$, we have from \eqref{pOP}
\begin{equation}
\label{monicOP}
\int_D \tilde{p}_n(z)\overline{\tilde{p}_m(z)} w(z)\,\d{A}(z)=\delta_{n,m}\tilde{h}_n\ .
\end{equation} 
It is well known (see e.g. \cite{Mehta}) that the normalisation constant in \eqref{eq:Zev} can be expressed in terms of these norms as
\begin{equation}
\label{eq:ZN}
\mathcal{Z}_N=\int_{D^N} \left|\Delta_N(z)\right|^2 
\prod_{i=1}^N w(z_i)\d{A}(z_i) = N!\prod_{j=0}^{N-1}\tilde{h}_j\ .
\end{equation}

\begin{rem}\label{rem:sel}
{\bf Selberg integrals}.
For our Gegenbauer polynomials with weight function $w(z)\d{A}=(1+\alpha)(1-h(z))^\alpha\d{A}=\d{A}_\alpha$, we have for 
 the  monic polynomials \eqref{eq:Cmonic}
 \begin{equation}
\tilde{p}^{(\alpha)}_n(z)=\frac{n!c^n}{2^n (1+\alpha)_{n}}C_{n}^{(1+\alpha)}(z/c)\ ,
 \end{equation}
 with orthogonality relation 
 \begin{equation}
 \label{monicC}
\Eint \tilde{p}^{(\alpha)}_n(z)\overline{\tilde{p}^{(\alpha)}_m(z)}\,\d A_\alpha{(z)}=\delta_{n,m}\tilde{h}_n ^{(\alpha)} 
\ ,
 \end{equation}
 and squared norms
 \begin{align} 
\tilde{h}_n ^{(\alpha)}
&=\frac{n!^2c^{2n}}{2^{2n} (1+\alpha)_{n}^2} \frac{1+\alpha}{1+\alpha+m}C_n^{(1+\alpha)}\le\frac{a^2+b^2}{a^2-b^2}\ri\delta_{n,m}
\nonumber\\
&=
\left(\frac{c}{2}\right)^{2n}
\frac{\sqrt{\pi}\Gamma(2+\alpha)\Gamma(2+2\alpha+n)\Gamma(n+1)}{2^{2\alpha+1}\Gamma(\alpha+\frac32)\Gamma(1+\alpha+n)\Gamma(2+\alpha+n)}
 \ F\left(2+2\alpha,-n;\alpha+\frac32;\frac{-b^2}{c^2}\right) .
\nonumber\\
\label{tildeh}
\end{align}
Here, we have used the representation \cite[8.932.1]{Grad} of Gegenbauer polynomials in terms of Gau{\ss}' hypergeometric function,
\be\label{eq:C-Frel}
C_n^{(1+\alpha)}(t)=\frac{\Gamma(2+2\alpha+n)}{\Gamma(n+1)\Gamma(2+2\alpha)}
\ F\left(2+2\alpha,-n;\alpha+\frac32;\frac{1-t}{2}\right) .
\ee
Consequently we obtain the following Selberg integral in the complex plane
\begin{eqnarray}
&&\int_{E^N}
\left.\left|\Delta_N(z)\right|^\beta 
\prod_{i=1}^N \left(1-\frac{1}{a^2}\Re(z_i)^2-\frac{1}{b^2}\Im(z_i)^2\right)^\alpha\frac{\d^2z_i}{\pi ab}\ \right|_{\beta=2}\nonumber\\
  =&& N!\frac{\pi^{\frac{N}{2}}N!\Gamma(1+\alpha)^N}{2^{(2\alpha+1)N}\Gamma(\alpha+\frac32)}
 \left(\frac{c}{2}\right)^{N(N-1)} 
  \prod_{n=0}^{N-1} 
\frac{\Gamma(2+2\alpha+n)\Gamma(n+1)}{\Gamma(1+\alpha+n)\Gamma(2+\alpha+n)}
\nonumber\\
&&\quad\quad\quad\quad\quad\quad\quad\quad\quad\quad\quad\quad\quad\quad\quad\times\ F\left(2+2\alpha,-n;\alpha+\frac32;\frac{-b^2}{c^2}\right)  
\ ,
\label{Zev}
\end{eqnarray}
after using the doubling formula for the $\Gamma$-function.
This one-parameter family can be analytically continued in $\alpha$. Of course for general $\alpha\in\mathbb{C}$ it will no longer be positive and can no longer be interpreted as a normalisation constant. It is an open problem 
how this result could be  extended to arbitrary $\beta\in\mathbb{C}$. 
\end{rem}
Let us return to our example for a set of orthogonal polynomials on the weighted ellipse $E$, with a recursion relation deeper that three steps. Therefore, we will apply 
the following theorem proved in \cite{akever},
which generalises Christoffel's Theorem for polynomials on $\mathbb{R}$: 
\begin{thm}\label{Thm:AV}
Let $\{v_i;\, i=1,\ldots,K\}$ and $\{u_i;\, i=1,\ldots,L\}$ be two
sets of complex numbers which are pairwise distinct among each set.
Without loss of generality we assume $K\geq L\geq0$, where the empty set
permitted.  Then the following
statement holds\footnote{The following notation is understood:
$\Delta_0(x)=\Delta_1(x)=1$ and $\prod_{i=N}^{M\leq N-1}h_i=1$.}:
\be
\left\langle  
\prod_{k=1}^N\left[\prod_{i=1}^K(v_i-z_k) \prod_{j=1}^L(\bar{u}_j-\bar{z}_k)\right]
\right\rangle_{N,w}
\ =\  \frac{\prod_{i=N}^{N+K-1}\tilde{h}_i^{\frac12}\ 
\prod_{j=N}^{N+L-1} \tilde{h}_j^{\frac12}}{\Delta_K(v)\ \Delta_L(\bar{u})}
\det_{1\leq l,m\leq K}[\ {\mathcal B}(v_l,\bar{u}_m)\ ] \ ,
\label{Th}
\ee
with matrix 
\be
{\mathcal B}(v_l,\bar{u}_m) \ \equiv\ \left\{ 
\begin{array}{ccl}
\kappa_{N+L}(v_l,\overline{u_m}):=\sum_{i=0}^{N+L-1}p_i(v_l)\overline{p_i(u_m)} & \mbox{for} & m=1,\ldots,L \\
&&\\
p_{N+m-1}(v_l)                           & \mbox{for} & m=L+1,\ldots,K \\
\end{array}
\right. .
\label{Ddef}
\ee
The monic polynomials $\tilde{p}_n(z)$ are  orthogonal w.r.t $w(z)$, with squared norms $\tilde{h}_n$ and  ${p}_n(z)=\tilde{p}_n(z)/\sqrt{\tilde{h}_n}$.
 \end{thm}

The multiplication operation on a sequence of polynomials  can be explicitly computed, using the above Theorem \ref{Thm:AV} for $K=2$ and $L=1$, as given in \cite{akever}. 
Following the Heine formula \eqref{Heine}, the polynomials $\{P^{(1)}_n\}_{n\in\mathbb{N}}$ orthogonal w.r.t. 
$|v-z|^2w(z)$ can be expressed in terms of the polynomials $p_n$ orthogonal with respect to $w(z)$. They are reading in monic normalisation 
\begin{eqnarray}
\tilde{P}^{(1)}_N(z) &=& 
\left\langle\, \prod_{i=1}^N(z-z_i)\, \right\rangle_{N,|v-\cdot|^2w} 
=\frac{\left\langle\prod_{i=1}^N(z-z_i)|v-z_i|^2\right\rangle_{N,w}}{ \left\langle\, \prod_{i=1}^N|v-z_i|^2\, \right\rangle_{N,w} }
\nonumber\\
&=&h_{N+1}^\frac12\frac{\ka_{N+1}(z,\bar{v})
p_{N+1}(v)-\ka_{N+1}(v,\bar{v})p_{N+1}(z)}{(v-z)
\ka_{N+1}(v,\bar{v})}  \ .
\label{OPL=1}
\end{eqnarray}
Their respective squared norms $\tilde{h}_N^{(1)}$ are not difficult to compute, using the orthonormality of the underlying polynomials $\tilde{p}_n$ \eqref{monicOP}:
\begin{eqnarray}
\tilde{h}_N^{(1)}&=& \int \tilde{P}^{(1)}_N(z) \overline{\tilde{P}^{(1)}_N(z)}\ |v-z|^2 w(z) \d{A}(z)\nonumber\\
&=& \frac{h_{N+1}}{\kappa_{N+1}(v,\bar{v})}\left( \kappa_{N+1}(v,\bar{v})|\tilde{P}^{(1)}_N(v)|^2+\kappa_{N+1}(v,\bar{v})^2\right) \nonumber\\
&=& \frac{h_{N+1}\kappa_{N+2}(v,\bar{v})}{\kappa_{N+1}(v,\bar{v})}\ .
\end{eqnarray}
This leads to the orthonormal polynomials
\begin{equation}\label{OP-norec}
{P}^{(1)}_N(z) = \frac{\ka_{N+1}(z,\bar{v})
p_{N+1}(v)-\ka_{N+1}(v,\bar{v})p_{N+1}(z)}{(v-z)
\sqrt{\ka_{N+1}(v,\bar{v})  \ka_{N+2}(v,\bar{v})  }} \ .
\end{equation}
The next step is to show that the Fourier coefficients of
\be
{z}P^{(1)}_N(z)  \ = \ \sum_{l=0}^{N+1} c_{l,N} \,P^{(1)}_l(z)
\ee
 are (in general) non-zero for $l\leq n-2$ for our example, when we choose $w(z)$ to be the Gegenbauer weight function, and thus the polynomials to be $\tilde{p}^{(\alpha)}_n$ from \eqref{monicC}, with squared norms \eqref{tildeh}.
Here, we may use that the orthonormalised Gegenbauer polynomials \eqref{bin} in the complex plane also satisfy a {\it three-term recurrence} relation \eqref{threeterm}, reading
\begin{equation}
\label{3t}
{z}\,p_n^{(\alpha)}(z)\ =\ a_{n+1}p^{(\alpha)}_{n+1}(z)+b_{n}p^{(\alpha)}_{n-1}(z)\ ,
\end{equation}
with 
\begin{equation}
\label{3tt}
a_{n+1} = \frac{c(n + 1)}{2 (n + \alpha + 1)}\sqrt{\frac{h_{n+1}}{h_{n}}}\ ,\quad 
b_n= \frac{c(n + 2\alpha+1)}{2 (n + \alpha + 1)}\sqrt{\frac{h_{n-1}}{h_{n}}}\ .
\end{equation}
Here, we use the definition from \eqref{binnorm} for the squared norms $h_n$ of the (un-normalised, non-monic) Gegenbauer polynomials. 
Notice that in contrast to the recursion for orthonormal Gegenbauer polynomials on the real line, the recurrence \eqref{3t} is not symmetric, $a_n\neq b_n$. This is due the difference in norm for $[-1,1]$ and $E$.
From now on we will use the following notation for $\ka_{i+1}(v,\bar{v}):=\ka_{i+1}$.
A simple  calculation implies that the coefficients 
\begin{equation}
c_{l,n}\ =\
\Eint z\,P_n^{(1)}(z) \overline{ P_l^{(1)}(z) }\ |v-z|^2\d A_\alpha(z)
\end{equation}
are given by
\begin{align}
c_{l,n}
=&\frac{1}{\sqrt{\ka_{n+1}\ka_{n+2}\ka_{l+1}\ka_{l+2}}}
\left[\left(\sum_{k=1}^la_kp_{k}^{(\alpha)}(v)p_{k-1}^{(\alpha)}(\bar{v})
-\sum_{k=0}^{min\{l,n-1\}}b_{k+1}p_{k}^{(\alpha)}(v)p_{k+1}^{(\alpha)}(\bar{v})
\right)
\right.\nonumber\\
&\times p_{l+1}^{(\alpha)}(\bar{v})p_{n+1}^{(\alpha)}({v})
-\Big(a_{l+1}p_{l}^{(\alpha)}(\bar{v})\Theta(n-l)-b_{l+2}p_{l+2}^{(\alpha)}(\bar{v})
\Theta(n-l-2)\Big)\ka_{l+1}p_{n+1}^{(\alpha)}({v})
\nonumber\\
&+\ka_{n+1}\ka_{n+2}a_{n+2}\delta_{n+1,l}-\ka_{n+1}b_{n+1}p_{n}^{(\alpha)}({v})
p_{l+1}^{(\alpha)}(\bar{v})\Theta(l-n)
+\ka_{n+1}\ka_{n}a_{n+1}\delta_{n-1,l}
\Bigg],\nonumber\\
\end{align}
where we have used the recursion \eqref{3t} and 
introduced the step function
\begin{equation}
\label{heavi}
\Theta(x):= 
\left\{
\begin{array}{ll}
1& \text{for}\;  x\geq0\ ,\\
&\\
0
& \mbox{for}\;  x<0\ .\\
\end{array}
\right.\\\ 
\end{equation}
If we only restrict ourselves to those indices $l\leq n-2$ which spoil the three-step recurrence, the remaining terms are simplified considerably and we obtain 
\begin{eqnarray}
c_{l\leq n-2,n}
&=&\frac{p_{n+1}^{(\alpha)}({v})}{\sqrt{\ka_{n+1}\ka_{n+2}\ka_{l+1}\ka_{l+2}}}
\left[\left(v\ka_{l}+b_lp_{l-1}^{(\alpha)}({v})p_{l}^{(\alpha)}(\bar{v})+ b_{l+1}p_{l}^{(\alpha)}({v})p_{l+1}^{(\alpha)}(\bar{v})\right)p_{l+1}^{(\alpha)}(\bar{v})
\right.
\nonumber\\
&&\quad\quad\quad\quad\quad\quad\quad\quad\quad\left.-\left(
a_{l+1}p_{l}^{(\alpha)}(\bar{v})+ b_{l+2}p_{l+2}^{(\alpha)}(\bar{v})\right) \ka_{l+1}
\right],
\label{clnrest}
\end{eqnarray}
which in general does not vanish for all $l$ down to zero.

Let us first check that we recover the three-term recurrence in the real limit $b\to0$, where we have to show  that indeed $c_{l\leq n-2,n}=0$ in this limit. 
When $b=0$ and the corresponding normalisation constants are understood as $h_n=h_n(a,0)$, the recursion \eqref{3tt} becomes symmetric,  $a_n=b_n$, - as it it known for Gegenbauer polynomials on $[-1,1]$ \cite{NIST}, cf. Remark \ref{reallim}. We thus obtain for the bracket in \eqref{clnrest} at $b=0$
\begin{eqnarray}
&&p_{l+1}^{(\alpha)}(\bar{v})\left(
v\kappa_l(v,\bar{v})
+b_lp_{l-1}^{(\alpha)}({v})p_{l}^{(\alpha)}(\bar{v})+ b_{l+1}p_{l}^{(\alpha)}({v})p_{l+1}^{(\alpha)}(\bar{v})
-\bar{v}\kappa_{l+1}(v,\bar{v})
\right)\no\\
&&=p_{l+1}^{(\alpha)}(\bar{v})\left(
\sum_{i=0}^{l-1}b_{i+1}p_{i+1}^{(\alpha)}({v})p_{i}^{(\alpha)}(\bar{v})   +\sum_{i=0}^{l+1}b_{i}p_{i-1}^{(\alpha)}({v}) p_{i}^{(\alpha)}(\bar{v})
\right.
\no\\
&&\qquad\qquad\left. -
\sum_{i=0}^{l}p_{i}^{(\alpha)}({v})
\left(b_{i+1}p_{i+1}^{(\alpha)}(\bar{v})   +b_{i}p_{i-1}^{(\alpha)}(\bar{v}) \right) 
\right)
\no\\
&&=0\ .
\end{eqnarray}
Here we have used the notation $p_{-1}^{(\alpha)}=0$, and after relabelling the sums, they can be seen to cancel.

To see that the expression \eqref{clnrest} is non vanishing in general for $b>0$, we consider the leading coefficient of \eqref{clnrest} as a polynomial in $\bar{v}$, which is of degree $2l+2$. We thus have to focus on 
\begin{equation}
\left(b_{l+1}p_{l+1}^{(\alpha)}(\bar{v})^2-
b_{l+2}p_{l}^{(\alpha)}(\bar{v})p_{l+2}^{(\alpha)}(\bar{v})\right) p_{l}^{(\alpha)}({v})\ .
\end{equation}
Because the polynomials of degree $l$ and $l+1$ do not have common zeros, it is sufficient to consider the leading coefficients inside the bracket, which read
\begin{eqnarray}
&&\frac{c(l+3+2\alpha)}{2(l+2+\alpha)}\sqrt{\frac{h_l}{h_{l+1}}}\frac{1}{h_{l+1}}\left( \frac{2^{l+1}\Gamma(2+l+\alpha)}{\Gamma(1+\alpha)(l+1)!c^{l+1}}\right)^2
\nonumber\\
&&- \frac{c(l+3+2\alpha)}{2(l+3+\alpha)}\frac{\sqrt{h_{l+1}}}{h_{l+2}}\frac{1}{\sqrt{h_{l}}}
\frac{2^{2l+2}\Gamma(1+\alpha+l)\Gamma(3+\alpha+l)}{\Gamma(1+\alpha)^2l!(l+2)!c^{2l+2}}\ ,
\label{coeff}
\end{eqnarray}
upon using \eqref{3tt} and \eqref{kappa}. Inserting \eqref{binnorm} and recalling \cite{NIST}
\begin{equation}
C_l^{(1+\alpha)}(1) =\frac{\Gamma(2+2\alpha+l)}{\Gamma(2+2\alpha)l!}\ ,
\end{equation}
it can be shown that \eqref{coeff} vanishes only if the following equality holds:
\be \label{eq1}
\frac{\left(C_{l+1}^{(1+\alpha)}\le x\ri\right)^2}{
\left(C_{l+1}^{(1+\alpha)}\le1\ri\right)^2}
-\frac{C_l^{(1+\alpha)}\le x\ri C_{l+2}^{(1+\alpha)}\le x\ri}{C_l^{(1+\alpha)}\le1\ri C_{l+2}^{(1+\alpha)}\le1\ri} =0\ ,
\ee
where 
\be
x=\frac{a^2+b^2}{a^2-b^2}\ .
\ee
The expression on the left hand side of \eqref{eq1}, usually denoted by $\Delta_n(x)$, it's know as Tur\'an determinant. By  \cite[Theorem 1]{Szwarc} $\Delta_n(x)=0$ if and only if $x=\pm1$.
Thus $c_{l,n}\equiv 0$ for $0\leq l \leq n-2$ in the limit $b \rightarrow 0$, that is when $x\to1$, which brings us back to the real line with a three-step recursion. For $x>1$  all 
Fourier coefficients are non-vanishing, $c_{l,n}\neq0$ for $0\leq l \leq n-2$, in our example of polynomials \eqref{OP-norec}  and no {\it finite-term recurrence} exists.


\begin{appendix}
\section{Proof of Lemma \ref{COP-lemma} for odd polynomials}
\label{La-odd}

In this appendix we collect the relevant formulae for the proof that \eqref{pd} holds when $m=2n+1$ and $j=2l+1$ are both odd. 
We begin by expressing the odd Gegenbauer polynomials in terms of a Gau{\ss} hypergeometric function, see e.g. \cite[8.932.3]{Grad}
\begin{eqnarray} 
C_{2n+1}^{(1+\alpha)}\left(\frac{z(r,\theta)}{c}\right)&=&
\frac{2z(r,\theta)}{c}\frac{(-1)^n\Gamma(n+2+\alpha)}{\Gamma(n+1)\Gamma(1+\alpha)}
F\left(-n, n+\alpha+2;\frac32; \frac{z(r,\theta)^2}{c^2}\right)
\no\\
&=&\frac{(-1)^n}{\G(1+\alpha)n!}\sum_{p=0}^{n}\sum_{k=0}^{2p+1}(-1)^p{n \choose p}{2p+1 \choose k}\frac{\G(2+\alpha +n+p)\G(p+1)}{\G(2p+2)}
\no\\
&&\qquad\qquad\qquad\times
r^{2p+1} R^{2(k-p)-1} e^{i\theta(2(k-p)-1)}
\no\\
&=&\frac{(-1)^n}{\G(1+\alpha)n!}\sum_{p=0}^{n}\sum_{k=0}^{2p+1}(-1)^p{n \choose p}{2p+1 \choose k}\frac{\G(2+\alpha +n+p)\G(p+1)}{\G(2p+2)}
\no\\
&&\qquad\qquad\qquad\times
r^{2p+1} R^{2(p-k)+1} e^{i\theta(2(p-k)+1)},
\label{eq:par2}
\end{eqnarray}
where we have used again the parametrisation \eqref{change}, giving two equivalent representations to be used. Likewise we obtain for the odd powers of the conjugated variable
\begin{eqnarray}
\left(\frac{\overline{z(r,\theta)}}{c}\right)^{2l+1}&=& 
\le \frac{r}{2}\ri^{2l+1} 
\left[\sum_{k=1}^{l+1}{2l+1 \choose k+l}R^{1-2k}e^{i\theta (2k-1)}+\sum_{k=1}^{l+1}{2l+1 \choose k+l}R^{2k-1}e^{i\theta(1- 2k)}\right].
\no\\
 \label{bin2}
\end{eqnarray}
The radial integral \eqref{radial} can be readily used, and we obtain for the angular integration
\begin{eqnarray}
&&\Eint C_{2n+1}^{(1+\alpha)}\left(\frac{z}{c}\right)\left(\frac{\overline{z}}{c}\right)^{2l+1}\,\d A_\alpha{(z)}=
\nonumber\\ 
&&= \frac{(1+\alpha)(-1)^n}{2^{2l+2}\pi}\sum_{k'=1}^{l+1}\sum_{p=0}^{n}\sum_{k=0}^{2p+1}{2l+1 \choose k'+l}\frac{(-1)^p\G(2+\alpha +n+p)\G(2+l+p)}{(n-p)!k!(2p-k+1)!\G(3+\alpha+l+p)}
\no\\
&&\qquad \qquad\qquad\qquad\times  
R^{2(k-p+k'-1)} \int_{0}^{2\pi}\d{\theta}\,e^{2i\theta(k-p-k')}
\no\\
&&\quad+  
\frac{(1+\alpha)(-1)^n}{2^{2l+2}\pi}\sum_{k'=1}^{l+1}\sum_{p=0}^{n}\sum_{k=0}^{2p+1}{2l+1 \choose k'+l}\frac{(-1)^p\G(2+\alpha +n+p)\G(2+l+p)}{(n-p)!k!(2p-k+1)!\G(3+\alpha+l+p)}
\no\\
&&\qquad \qquad\qquad\qquad\times  R^{2(p-k-k'+1)} \int_{0}^{2\pi}\d{\theta}\,e^{2i\theta(p-k+k')}.
\label{thetaint2}
\end{eqnarray}
Let us evaluate the first triple sum, where we have $k=p+k'$ due to the angular integration, and because of $k\leq 2p+1$ this implies $k'\leq p+1$.  We thus obtain for it
\begin{eqnarray}
&&\frac{(1+\alpha)(-1)^n}{2^{2l+1}}\sum_{k'=1}^{l+1}{2l+1 \choose k'+l}\!\!
\sum_{p=k'-1}^{n}\frac{(-1)^p\G(2+\alpha +n+p)\G(2+l+p)}{(n-p)!(k'+p)!(p+1-k')!\G(3+\alpha+l+p)}R^{4k'-2}
\no\\
&&:=\frac{(1+\alpha)(-1)^n}{2^{2l+1}}\sum_{k'=1}^{l+1}{2l+1 \choose k'+l}b_{k'}R^{4k'-2}.
\no\\
\label{thirdsum2}
\end{eqnarray}
It is a polynomials in $R$ of degree $4l+2$, and we have to show that all its coefficients $b_{k'}=b_{k'}(n,l)$ vanish for $l<n$. 
The second triple sum in \eqref{thetaint2} agrees with \eqref{thirdsum2} with $R\to R^{-1}$, due to $k=p+k'$ and $k'\leq p+1$ from the angular integration there.  
The coefficients $b_{k}$ can again be rewritten as an integral. After shifting the summation index we have
\begin{eqnarray}
b_k&=&\frac{(-1)^{k-1}}{(n+1-k)!}\,\sum_{p=0}^{n+1-k}(-1)^p{n+1-k \choose p}\frac{\G(1+\alpha +n+k+p)\G(1+l+k+p)}{\G(2k+p)\G(2+\alpha+l+k+p)}
\no\\
&=&\frac{(-1)^{k-1}}{(n+1-k)!}\,\int_{0}^{1}\d{x}\frac{x^{l+k}(1-x)^{\alpha}}{\G(1+\alpha)}\sum_{p=0}^{n+1-k}(-1)^p{n+1-k \choose p}\frac{\G(1+\alpha +n+k+p)}{\G(2k+p)}x^p
\no\\
&=&\frac{(-1)^{k-1}\G(1+\alpha+n+k)}{(n+1-k)!\G(1+\alpha)\G(2k)}\int_{0}^{1\!}\d{x}x^{l+k}(1-x)^{\alpha}F(-n-1+k, 1+\alpha+n+k; 2k; x).
\no\\
\end{eqnarray}
The very same steps as in the proof for the even polynomials allow us to manipulate the remaining integral as follows:
\begin{equation}\label{prueba4impar}
\begin{split}
&\int_{0}^{1}\,\d{x}\, x^{l+k}(1-x)^{\alpha}F(-(n+1-k), 1+\alpha+n+k; 2k; x)\\
&=\lim_{\varepsilon\to0}\int_{0}^{1}\,\d{x}\, x^{l+k+\varepsilon}(1-x)^{\alpha}F(-(n+1-k), 1+\alpha+n+k; 2k; x)\\
&=\lim_{\varepsilon\to0} \frac{\G(2k)\G(1+l+k+\varepsilon) \G(2+\alpha+n-k)\G(n-l-\varepsilon)}{\G(1+k+n)\G(3+\alpha+n+l+\varepsilon)\G(-\varepsilon-(l+1-k))}\\
&=\lim_{\varepsilon\to0} \frac{(-1)^{l-k}\G(2k)\G(1+l+k+\varepsilon) \G(2+\alpha+n-k)\G(l+2-k+\varepsilon)}{\pi\G(1+k+n)\G(3+\alpha+n+l+\varepsilon)}\\
&\qquad \qquad \times \G(n-l-\varepsilon)\sin(\pi\varepsilon)\ .
\end{split}
\end{equation}
Together with \eqref{sinepsilon} this establishes the orthogonality of the odd polynomials \eqref{thetaint2}.
In order to compute the norms for the odd polynomials we summarise the above results for the coefficients 
\begin{equation}
b_k(n,n)=\frac{(-1)^n\Gamma(1+\alpha+n+k)\Gamma(2+\alpha+n-k)}{\Gamma(1+\alpha)\Gamma(3+\alpha+2n)}\ ,
\label{bnk}
\end{equation}
which has to be inserted into \eqref{thirdsum2}, and the corresponding equation with $R\to R^{-1}$. We obtain at $n=l$
\begin{eqnarray}
&&\sum_{k'=1}^{n+1}{2n +1\choose k'+n}b_{k'}(n,n)R^{\pm(4k'-2)}\no\\
&=& 
\frac{(-1)^n(2n+1)!}{\Gamma(1+\alpha)\Gamma(2n+\alpha+3)}
\sum_{k=n+1}^{2n}\frac{\Gamma(1+\alpha+k)\Gamma(2+\alpha+2n-k)}{(2n+1-k)!k!}R^{\mp(4(n-k)+2)}
\no\\
&=& 
\frac{(-1)^n(2n+1)!}{\Gamma(1+\alpha)\Gamma(2n+\alpha+3)}
\sum_{k=0}^{n}\frac{\Gamma(1+\alpha+2n+1-k)\Gamma(1+\alpha+k)}{k!(2n+1-k)!}R^{\pm(4(n-k)+2)},\no\\ 
\label{bnksum}
\end{eqnarray}
relabelling the sum twice. These two sums with $R^\pm$ can be inserted into \eqref{thetaint2} to give a single sum
\begin{eqnarray}
&&\Eint C_{2n+1}^{(1+\alpha)}\left(\frac{z}{c}\right)\left(\frac{\overline{z}}{c}\right)^{2l+1}\,\d A_\alpha{(z)}\no\\
&&=\frac{\delta_{n,l}(1+\alpha)(2n+1)!}{2^{2n+1}\Gamma(1+\alpha)\Gamma(2n+\alpha+3)}\sum_{k=0}^{2n+1}\frac{\Gamma(1+\alpha+k)\Gamma(1+\alpha+2n+1-k)}{\Gamma(2n+1-k+1)\Gamma(k+1)}R^{4n+2-4k}.\no\\
\label{C2n+1}
\end{eqnarray}
This sum can be written as a single Gegenbauer polynomial, using its  invariance under $k\to2n+1-k$:
 \begin{eqnarray}
&=&\frac12 \sum_{k=0}^{2n+1}\frac{\Gamma(1+\alpha+k)\Gamma(1+\alpha+2n+1-k)}{\Gamma(2n+1-k+1)\Gamma(k+1)}
\left(R^{4n+2-4k}+R^{-(4n+2-4k)}\right)\no\\
 &=&\sum_{k=0}^{2n}\frac{\Gamma(1+\alpha+k)\Gamma(1+\alpha+2n+1-k)}{\Gamma(2n+1-k+1)\Gamma(k+1)}\cosh[(2n+1-2k)\ln(R^2)]\no\\
 &=&\Gamma(1+\alpha)^2 C_{2n+1}^{(1+\alpha)}\left(\frac{a^2+b^2}{a^2-b^2}\right)\ ,
 \label{Cabodd}
 \end{eqnarray}
where we used again \cite[18.5.11]{NIST} and \eqref{Rab}.
The leading power of the odd Gegenbauer polynomials can be read off  from the first line of \eqref{eq:par2},
\begin{equation}
C_{2l+1}^{(1+\alpha)}(x)=\frac{\Gamma(2l+2+\alpha)2^{2l+1}}{\Gamma(1+\alpha)(2l+1)!} x^{2l+1} + O(x^{2l-1})\ .
\end{equation}
Multiplying \eqref{thetaint2} with this factor and using that the lower powers vanish yields
\begin{eqnarray}
&&\Eint C_{2n+1}^{(1+\alpha)}\left(\frac{z}{c}\right)C_{2l+1}^{(1+\alpha)}\left(\frac{\overline{z}}{c}\right)\,\d A_\alpha{(z)}
=\delta_{2n+1,2l+1}
\frac{1+\alpha}{2n+\alpha+2}
 C_{2n+1}^{(1+\alpha)}\left(\frac{a^2+b^2}{a^2-b^2}\right).\no\\
\end{eqnarray}
It agrees with \eqref{OP} for odd indices. 

\section{Orthogonality of Chebyshev polynomials of second kind}
\label{app:Uproof}

For completeness we present an independent proof for the orthogonality of the Chebyshev polynomials of the second kind $U_n$ on the interior of the ellipse $E$ \eqref{ellipsedef},
\begin{equation}\label{normacheb}
\Eint U_{m}\left(\frac{z}{c}\right)U_{n}\left(\frac{\overline{z}}{c}\right)\,\d A{(z)}
=\frac{1}{1+n}U_n\le\frac{a^2+b^2}{a^2-b^2}\ri\delta_{n,m}\ .
\end{equation}
The argument of the proof is not new and it can be found in \cite[pag. 546]{Henrici}. It uses Stokes' Theorem (see e.g.  \cite{Hormander}), that we restate for the readers convenience. 

Let $G$ to be a bounded open set in $\mathds{C}$, such that the boundary $\partial G$ consists of a finite number
of $C^1$ Jordan curves. 
For any  $F \in C^{1}\le\overline{G}\ri$  Stokes' Theorem relates the integral over $G$ to that over its boundary $\partial G$:
\be
\int_G \overline{\partial}F(z)\d{A}(z)=\frac{1}{2i}\int_{\partial G}  F(z)\d z, \quad \overline{\partial}:=\frac{\partial}{\partial \overline{z}}\ .
\ee
In particular for $F(z)=f(z)\overline{g(z)}$ with $f,g$ analytic, we have
\be\label{Stokes}
\int_G\overline{ \partial }\left[f(z)\overline{g(z)}\right]\d{A}(z)=\int_Gf(z)\overline{g'(z)}\d{A}(z)=\frac{1}{2i}\int_{\partial G} f(z)\overline{g(z)} \d z\ .
\ee

\begin{proof}
To show \eqref{normacheb}, we can use the well-known formula \cite[18.9.21]{NIST} relating Chebyshev polynomials of the first $T_n$ and second kind
\be
T'_n(z)=nU_{n-1}(z)\ ,
\ee
for $n=1,2,\ldots$ We can thus rewrite the l.h.s of \eqref{normacheb} for any $n,m=0,1,\ldots$ to apply Stokes' Theorem
\begin{equation}\label{proofstep1}
\begin{split}
\Eint U_{n}\left(\frac{z}{c}\right)U_{m}\left(\frac{\overline{z}}{c}\right)\,\d A{(z)}&= \frac{c^2}{(n+1)(m+1)}\Eint T'_{n+1}\le\frac{z}{c}\ri\overline{T'_{m+1}\le\frac{z}{c}\ri}\d A(z)\\
&= \frac{c^2}{(n+1)(m+1)}\Eint \overline{\partial}\left[T'_{n+1}\le\frac{z}{c}\ri\overline{T_{m+1}\le\frac{z}{c}\ri}\right]\d A(z)\\
&= \frac{c^2}{(n+1)(m+1)}\frac{1}{2i}\int_{\partial E} T'_{n+1}\le\frac{z}{c}\ri\overline{T_{m+1}\le\frac{z}{c}\ri} \frac{\d z}{\pi ab}\ .\\
\end{split}
\end{equation} 
Next, we use the Joukowsky map  
\be\label{Joukowsky}
z(w)=\frac{1}{2}\le w+\frac{c^2}{w}\ri\ ,
\ee
that maps the circle $|w|=r$ of radius $r:=a+b$ onto the boundary $\partial E\ni z$ of the ellipse $E$.  The chain rule
\be\label{derivative} 
T'_{n+1}(z)=\frac{d}{dz}T_{n+1}(z)=\frac{d}{dw}T_{n+1}(z(w))\ \frac{dw}{dz}\ ,
\ee
allows us to rewrite
\begin{equation}\label{proofstep2}
\begin{split}
\int_{\partial E} T'_{n+1}\le\frac{z}{c}\ri\overline{T_{m+1}\le\frac{z}{c}\ri}\; \d z&= \int_{|w|=r} \frac{d}{dw}T_{n+1}(z(w)/c)\overline{T_{m+1}(z(w)/c)}\;\d w\ .\\
\end{split}
\end{equation} 
Note that the contribution from the Jacobian of the transformation \eqref{Joukowsky} just cancels the extra factor $\frac{dw}{dz}$ stemming from \eqref{derivative}. Furthermore, as we have stated already in \eqref{TJ}, it is well known \cite{Mason} that 
\begin{equation}
T_{n+1}(z(w)/c)=\frac{1}{2}((w/c)^{n+1}+(c/w)^{n+1})
\end{equation}
 holds for the Joukowsky map \eqref{Joukowsky}. Therefore
\begin{equation}\label{proofstep3}
\begin{split}
\frac{d}{dw}T_{n+1}(z(w)/c)&= \frac{n+1}{2w}\left[ \le\frac{w}{c}\ri^{n+1}-\le\frac{c}{w}\ri^{n+1}\right]\\
\end{split}
\end{equation} 
allows us to exploit the orthogonality on the circle
\begin{equation}
\int_{|w|=r} w^a \bar{w}^b \frac{\d w}{w}= ir^{a+b}\int_0^{2\pi}e^{i\theta(a-b)}\d \theta=2\pi i r^{2a}\delta_{a,b}\ ,
\end{equation}
as follows:
\begin{eqnarray}
&&\int_{\partial E} T'_{n+1}\le\frac{z}{c}\ri\overline{T_{m+1}\le\frac{z}{c}\ri}\; \d z\nonumber\\
&&=\frac{n+1}{4} \int_{|w|=r} \left[ \le\frac{w}{c}\ri^{n+1}-\le\frac{c}{w}\ri^{n+1}\right] \overline{\left[ \le\frac{w}{c}\ri^{m+1}+\le\frac{c}{w}\ri^{m+1}\right]}\frac{\d w}{w}\nonumber\\
&&=\frac{i\pi(n+1)}{2} \left[ \le\frac{r}{c}\ri^{2n+2}- \le\frac{c}{r}\ri^{2n+2}\right]\delta_{n,m}\ .
\label{proofstep4}
\end{eqnarray}
In this form the orthogonality is stated in \cite{Henrici}. To arrive at the right hand side of \eqref{normacheb} 
we use \eqref{C2n+1simple} and \eqref{C2nsimple}, together with $U_n=C_n^{(1)}$.
\end{proof}

\section{Alternative 
proof of Lemma \ref{COP-lemma}}
\label{La-proof2}
In the proof presented here we do not need to distinguish between Gegenbauer polynomials with even and odd parity. The main assumption to be made here is that the orthogonality of the Chebyshev polynomials of the second kind $U_n(x)$ holds on the unweighted ellipse, as shown in \eqref{normacheb} in the previous Appendix \ref{app:Uproof}. 
Using $U_n(x)=C_n^{(1)}(x)$, we only need to show that from \eqref{normacheb} follows that
\begin{equation}
\Eint C_{m}^{(1+\alpha)}\left(\frac{z}{c}\right)\overline{z}^j\,\d A_{\alpha(z)}=0\ ,\ \ j=0,1,\ldots,m-1\ ,
\label{LaB}
\end{equation}
holds for  $\alpha>-1$.
The computation of the norms  however, is much more involved in this approach and we will not present it here. It leads to the same result as given in Lemma \ref{COP-lemma}.

The general Gegenbauer polynomials can be explicitly written in the following form,
as can be seen for example from the representations in terms of a hypergeometric function  (see \eqref{eq:par} and \eqref{eq:par2}),
\begin{equation}
\label{Cexpand}
C^{(1+\alpha)}_n(z) = \sum_{j=0}^n \kappa^n_j(\alpha) z^j\ ,
\end{equation}
with the coefficients reading
\begin{equation}
\label{kappa}
\kappa^n_j(\alpha)=  \left\{ \begin{array}{ll} 
 \frac{\displaystyle (-1)^{(n - j)/2} 2^j \Gamma(\alpha + 1+(n + j)/2)}{
\displaystyle \Gamma(\alpha + 1) \Gamma(j + 1) \Gamma(1+(n -j)/2)},
 &\mbox{for}\ \  n - j \ {\rm even},\\ 
0, & \mbox{for}\ \ n - j \ {\rm  odd}. \end{array} \right. 
\end{equation}
This immediately implies the following relation
\begin{equation}
\label{ka}
\kappa^n_j(\alpha)=\frac{\Gamma(\alpha + 1+(n + j)/2)}{\Gamma(\alpha + 1)\Gamma(1+(n +j)/2)}\kappa^n_j(0)\ ,
\end{equation}
between the expansion coefficients for general Gegenbauer  and Chebyshev polynomials of the second kind ($\alpha=0$).
The former satisfy the following three-term recurrence relation
\begin{equation}
\label{threeterm}
z C^{(1+\alpha)}_n(z) = \frac{n + 1}{2( n + \alpha + 1)} C^{(1+\alpha)}_{n+1}(z) + \frac{n + 2 \alpha + 1}{2( n + \alpha + 1)} C^{(1+\alpha)}_{n-1}(z), \ \ \ n = 1,2,3,\ldots\ .
\end{equation}
Let us recall our notation for the inner product \eqref{innerproduct}, 
\begin{equation}
\label{innerproductB}
\langle f,g\rangle_\alpha := \Eint f(z)\overline{g(z)}\,\d A_\alpha{(z)} \ .
\end{equation}
From  (\ref{normacheb}) we immediately have for arbitrary integers $m$ and $j$, satisfying $m > j \geq 0$, that
\begin{equation}
\langle C^{(1)}_m(z/c),z^j \rangle_0 = 0\ .
\end{equation}
Using the three-term recurrence relation \eqref{threeterm} for $\alpha=0$, we see that $z^l C^{(1)}_m(z/c)$ can be expanded in terms of $C^{(1)}_k(z/c)$ 
with $m-l \leq k \leq m+l$. It thus follows that
\begin{equation}
\label{C1zbz}
\langle C^{(1)}_m(z/c), z^j {\bar z}^l \rangle_0 = 
\langle z^l C^{(1)}_m(z/c), z^j \rangle_0 =0
\end{equation}
holds, for $j \geq 0$, $l \geq 0$ and $j + l < m$. Our goal is to use this relation and to prove the orthogonality \eqref{LaB} by relating inner products, with 
general $\alpha>-1$ to those with $\alpha=0$. 
In particular we can generalise the above statement \eqref{C1zbz} to the following
\begin{lem}
For an arbitrary positive integer $m$
\begin{equation}
\langle C^{(1+\alpha)}_m(z),z^j {\bar z}^l\rangle_{\alpha} = 0
\label{Corth}
\end{equation}
holds for $\alpha>-1$, given that $j \geq 0$, $l \geq 0$ and $j + l <m$.
\end{lem}
Obviously \eqref{LaB} then follows by choosing $l=0$.
\begin{proof} 
As it was already noted, due to the invariance of the weight and domain under the reflection $z\to-z$, we have that 
\begin{equation}
\label{parity}
\langle z^p,z^q\rangle_\alpha\neq 0 \ \ \mbox{iff}\ \ p+q \ \ \mbox{even}\ ,
\end{equation}
that is when $p$ and $q$ have the same parity. Furthermore, we can relate the inner products of such monomials with general $\alpha>-1$ and with $\alpha=0$ as follows. The change of variables \eqref{change} decouples radial and angular integration and leads to 
\begin{eqnarray}
\langle z^p,z^q\rangle_\alpha &=& \frac{1+\alpha}{\pi}\int_0^1\d r\,r  \left(\frac{r}{2}\right)^{p+q}(1-r^2)^\alpha \!\int_0^{2\pi}\d \theta
(R e^{i\theta}+R^{-1}e^{-i\theta})^p(R^{-1} e^{i\theta}+Re^{-i\theta})^q\
\no\\
&=& \frac{\Gamma(2+\alpha)\Gamma(2+(p+q)/2)}{\Gamma(2+\alpha+(p+q)/2)} \langle z^p,z^q\rangle_0\ ,
\label{al0rel}
\end{eqnarray}
due to standard integrals
\begin{equation*}
\int_0^1\d r\,r^{p+q+1}(1-r^2)^\alpha= \frac{\Gamma(1+\alpha)\Gamma(1+(p+q)/2)}{\Gamma(2+\alpha+(p+q)/2)} \ .
\end{equation*}
We proceed to prove \eqref{Corth} by induction. For $m=1$, we can readily find
\begin{equation}
\langle C^{(1+\alpha)}_1(z/c), 1 \rangle_{\alpha} = \kappa^1_1(\alpha) \langle z/c,1 \rangle_{\alpha} = 0\ ,
\end{equation} 
which holds due to parity, see \eqref{parity}.
Now, suppose that the claim \eqref{Corth} holds for $m = 1,2,\cdots,k$. 
We will show 
\begin{equation}
\langle C^{(1+\alpha)}_{k+1}(z/c),z^j {\bar z}^l\rangle_{\alpha} = 0\ ,
\end{equation}
separately for (i) $j+l\leq k-2$, (ii) $j+l=k-1$ and (iii) $j+l=k$.

\begin{itemize}
\item[(i)] If $j,l \geq 0$ and $j + l+1 \leq k-1$, the induction assumption guarantees that
\begin{eqnarray}
\label{e1}
0 & = & c^{-1}\langle C^{(1+\alpha)}_k(z/c),z^j {\bar z}^{l+1}\rangle_{\alpha} 
 = \langle (z/c) \,C^{(1+\alpha)}_k(z/c),z^j {\bar z}^l \rangle_{\alpha}
\nonumber \\ 
& = & \frac{k+1}{2 (k + \alpha + 1)} \langle C^{(1+\alpha)}_{k+1}(z/c),z^j {\bar z}^l \rangle_{\alpha} 
+ \frac{k + 2 \alpha + 1}{2 (k + \alpha + 1)} \langle C^{(1+\alpha)}_{k-1}(z/c),z^j {\bar z}^l \rangle_{\alpha} 
\nonumber \\ 
& = & \frac{k+1}{2 (k + \alpha + 1)} \langle C^{1+\alpha}_{k+1}(z/c),z^j {\bar z}^l \rangle_{\alpha}. 
\end{eqnarray}
Here, we have used the recursion relation \eqref{threeterm} and in the second line again the induction assumption, to arrive at the claimed statement.

\item[(ii)] If $j,l \geq 0$ and $j + l = k-1$, we may directly use the expansion \eqref{Cexpand} to obtain 
\begin{eqnarray}
\label{e3}
\frac{1}{1+\alpha} \langle C^{(1+\alpha)}_{k+1}(z/c),z^j {\bar z}^l \rangle_{\alpha}
&=& \frac{1}{1+\alpha}\sum_{p=0}^{k+1} \kappa^{k+1}_p(\alpha) \langle (z/c)^p, z^j {\bar z}^l \rangle_{\alpha} 
\nonumber \\ 
& = & \sum_{p=0}^{k+1} \kappa^{k+1}_p(\alpha) \frac{\Gamma(\alpha + 1) \Gamma((k + p + 3)/2)}{\Gamma((k + 2 \alpha + p + 3)/2)}  
\langle (z/c)^p, z^j {\bar z}^l \rangle_0 
\nonumber \\ 
& = & \sum_{p=0}^{k+1} \kappa^{k+1}_p(0) \langle (z/c)^p, z^j {\bar z}^l \rangle_0 \no\\
&=& \langle C^{(1)}_{k+1}(z/c),z^j {\bar z}^l \rangle_0 = 0\ .
\end{eqnarray}
In the second step we have used the relation \eqref{al0rel}, to be able to relate to the known orthogonality \eqref{C1zbz} via \eqref{Cexpand}. 

\item[(iii)] If $j,l \geq 0$ and $j + l = k$, we see from \eqref{Cexpand} that  
\begin{equation}
\label{e2}
\langle C^{(1+\alpha)}_{k+1}(z/c),z^j {\bar z}^l \rangle_{\alpha} 
= \sum_{p=0}^{k+1} \kappa^{k+1}_p(\alpha) \langle (z/c)^p, z^j {\bar z}^l \rangle_{\alpha} = 0
\end{equation}
holds due to parity. This is because from \eqref{kappa} $k+1-p$ is even, whereas $p+j-l=p+k-2l$ is consequently odd, leading to a vanishing expectation value.
\end{itemize}
\end{proof}

\end{appendix}



\begin{thebibliography}{99}

\bibitem{akever}
Akemann, G.; Vernizzi, G. Characteristic Polynomials of Complex Random Matrix Models. Nucl. Phys. B 660:3 (2003) 532--556 [arXiv:hep-th/0212051].


\bibitem{ABe}
  Akemann, A.; Bender, M.
  Interpolation between Airy and Poisson statistics for unitary chiral
  non-Hermitian random matrix ensembles.
  J. Math. Phys. 51 (2010) 103524
  [arXiv:1003.4222 [math-ph]].
  
\bibitem{APh}  
Akemann, G.; Phillips, M.J.
Universality Conjecture for all Airy, Sine and Bessel Kernels in the Complex Plane
in Random Matrix Theory. In "Interacting Particle Systems, and Integrable Systems",  P. Deift and P. Forrester (eds.),
MSRI Publications, Volume 65 (2014) 1-24,
Cambridge University Press, ISBN-13: 978-1-107-07992-2 [arXiv:1204.2740 [math-ph]].



\bibitem{Conway}
Conway, J. B. A Course in Functional Analysis. Second edition, Springer, New York, 1990. 

\bibitem{PdF} Di Francesco, P.; Gaudin, M.; Itzykson, C.; Lesage, F.
Laughlin's wave functions, Coulomb gases and expansions of the discriminant.
  Int. J. Mod. Phys. A9 (1994) 4257 [hep-th/9401163].

\bibitem{EM}
 van Eijndhoven, S.J.L.; Meyers, J.L.H. New orthogonality relation for the Hermite polynomials and related Hilbert space. J. Math. Ana. Appl. 146 (1990) 89--98.

\bibitem{Peter-Selberg} Forrester, P.J.; Warnaar, S.O.
The importance of the Selberg integral. 
Bull. Amer. Math. Soc. (N.S.) 45 (2008), 489-534 [ arXiv:0710.3981 [math.CA]]. 

\bibitem{FKS} Fyodorov, Y.V.; Khoruzhenko, B.A.; Sommers, H.-J.
Almost-Hermitian Random Matrices: Eigenvalue Density in the Complex Plane.
Phys. Lett. A 226 (1997) 46--52  [arXiv:cond-mat/9606173];
Almost-Hermitian Random Matrices: Crossover from Wigner-Dyson to Ginibre eigenvalue statistics.
Phys. Rev. Lett. 79 (1997) 557--560  [arXiv:cond-mat/9703152].

\bibitem{FKS98} Fyodorov, Y.V.; Khoruzhenko, B.A.; Sommers, H.-J.
Universality in the random matrix spectra in the regime of weak non-Hermiticity.
Ann.\ Inst.\ Henri Poincar\'e
{ 68} (1998) 449--489 [arXiv:chao-dyn/9802025].

\bibitem{Grad}
Gradshteyn, I.S.; Ryzhik, I.M. Table of Integrals, Series, and Products.
A. Jeffrey and D. Zwillinger (eds.). Seventh edition, Academic Press, San Diego 2007.

\bibitem{Henrici}
Henrici, P. Applied and Computational Complex Analysis. Vol. 3, (1993) Wiley-Interscience, Volume in Pure and Applied Mathematics.

\bibitem{Hormander}
H\"{o}rmander, L. An introduction to complex analysis in several variables. Second edition, North-Holland, Amsterdam, 1979.

\bibitem{Karp} Karp, D. 
Square summability with geometric weight for classical orthogonal expansions. In
Advances in Analysis, Proc. 4th Int. ISAAC Conf., H. G. W. Begehr et al. (eds.), World Scientific, Singapore, 
(2005) 407--421
[arXiv:math/0604028 [math.CA]].

\bibitem{KS}
Khavinson, D.; Shapiro, H.S. Dirichlet's problem when the data is an entire
function. Bull. London Math. Soc. 24, no. 5 (1992) 456--468.

\bibitem{khavinson}
  Khavinson, D.; Stylianopoulos, N. Recurrence relations for orthogonal polynomials and algebraicity of solution of the Dirichlet problem. Int. Math. Series, 12 (2010) 219--228.



\bibitem{Lempert}
Lempert, L. Recursion for orthogonal polynomials on complex domains. In:
Fourier Analysis and Approximation Theory (Proc. Colloq., Budapest, 1976),
Vol. II, pp. 481--494, North-Holland, Amsterdam 1978.

\bibitem{Mason}
Mason, J.; Handscomb, D. Chebyshev Polynomials. Chapman and Hall/CRC Press, London, 2002. 

\bibitem{Mehta}
Mehta, M.-L. Random Matrices, Academic Press, Third edition, London
2004.


\bibitem{AKNP} Nagao, T.; Akemann, G.; Kieburg, M.;  Parra, I. 
Families of two-dimensional Coulomb gases on an ellipse: correlation functions and universality. Preprint arXiv:1905.07977 [math-ph] (2019), 41 pages.

\bibitem{NIST}
Olver, F.W.J. et al. (eds.), NIST Handbook of Mathematical Functions, Cambridge University Press, Cambridge,
2010.

\bibitem{Osborn} 
Osborn, J.C.
Universal results from an alternate random matrix model for QCD with a baryon chemical potential.
Phys. Rev. Lett. 93 (2004) 222001 [arXiv:hep-th/0403131].

\bibitem{PS} Putinar, M.; Stylianopoulos, N.S.
Finite-term relations for planar orthogonal polynomials. Compl. Ana. Op. Th. 1(3)
(2007) 447-456.

\bibitem{SCSS}
Sommers, H. J.; Crisanti, A.; Sompolinsky, H.; Stein, Y.  Spectrum of large random asymmetric matrices. Phys. Rev. Lett., 60(19), (1988) 1895--1898.

\bibitem{Szwarc}
Szwarc, R; Positivity of Tur\'an determinants for orthogonal polynomials, in ``Harmonic Analysis and Hypergroups" (Delhi, 1995), Birkh\"auser, Boston MA, 1998, pp. 165--182.

\bibitem{Szego}
Szeg\"o, G.  Orthogonal polynomials, 4th ed, American Mathematical Society, Colloquium Publications, vol. XXIII, Providence, R.I., 1975.


\bibitem{Walter}
Van Assche, W. Orthogonal polynomials in the complex plane and on the real line. In "Special functions, q-series and related topics", M.E.H. Ismail et al. (eds), Fields Inst. Commun. 14 (1997), 211--245.


\bibitem{ZS}
\.Zyczkowski, K.; Sommers, H.-J. 
Truncations of random unitary matrices. 
J. Phys. A: Math. Gen. 33
(2000) 2045--2057 [arXiv:chao-dyn/9910032].



  
  
\end{thebibliography}
\end{document}